%% file: cores_quasiconvex_hyperbolic_arxiv.tex
\documentclass[11pt,letterpaper]{amsart}
\usepackage[leqno]{amsmath}
\usepackage{amsmath,amsfonts,amsthm,amssymb}
\usepackage{epsfig}
\usepackage{graphicx}
\usepackage{subfigure}
\usepackage{cite,url}
\usepackage{hyperref}
\usepackage[left=1in,right=1in,top=1in,bottom=1in,foot=0.5in]{geometry}

\usepackage{epsfig}
\usepackage{color}
\usepackage{amsmath}
\usepackage{amssymb}
\newtheorem{lemma}{Lemma}
\newtheorem{proposition}{Proposition}
\newtheorem{theorem}{Theorem}
\newtheorem{corollary}{Corollary}

\theoremstyle{definition}

\newtheorem{remark}{Remark}

\newcommand{\kb}{{}^\kappa \mathcal{Q}}

\newcommand{\kqm}{{}^\kappa\mathcal{Q}}
\newcommand{\kqa}{{}\mathcal{Q}}
\newcommand{\kq}{{}^\kappa Q}

\usepackage{color}
\newcommand{\FFD}[1]{{\color{red} #1}} 
\newcommand{\commentout}[1]{}

\newcommand{\diam}{\mathop{\rm diam} }
\newcommand{\rad}{\mathop{\rm rad} }
\newcommand{\ecc}{\mathop{\rm ecc} }

\usepackage{color}

\begin{document}


\medskip

\vspace{1.2cm}
\centerline{\Large\bf  Core congestion is inherent in hyperbolic networks}

\vspace{10mm}
\centerline{Victor Chepoi$^{\small 1}$, Feodor F. Dragan$^{\small 2}$, and Yann Vax\`es$^{\small 1}$}

\medskip
\begin{small}
\medskip
\centerline{$^{1}$Laboratoire d'Informatique Fondamentale, Aix-Marseille Univ. and CNRS,}
\centerline{Facult\'e des Sciences de Luminy, F-13288 Marseille Cedex 9, France}

\centerline{\texttt{\{victor.chepoi, yann.vaxes\}@lif.univ-mrs.fr}}

\medskip
\centerline{$^{2}$Computer Science Department, Kent State University,}
\centerline{Kent, OH 44242, USA}
\centerline{\texttt{dragan@cs.kent.edu}}
\end{small}

\bigskip\noindent
{\footnotesize {\bf Abstract.}  We investigate the impact the negative curvature has on the traffic congestion in large-scale networks. We prove that every Gromov hyperbolic network $G$ admits a core, thus answering in the positive
a conjecture by Jonckheere, Lou, Bonahon, and Baryshnikov, Internet Mathematics, 7 (2011) which is based on the experimental observation by
Narayan and Saniee, Physical Review E, 84 (2011)  that real-world networks with small hyperbolicity have a core congestion. Namely, we prove that for every subset $X$  of 
vertices of
a graph with $\delta$-thin geodesic triangles (in particular, of a $\delta$-hyperbolic graph) $G$ there exists a vertex $m$ of $G$ such that the ball $B(m,4 \delta)$ of radius $4 \delta$ centered at $m$ intercepts at least one half of the
total flow between all pairs of vertices of $X$, where the flow between two vertices $x,y\in X$ is carried by geodesic (or quasi-geodesic) $(x,y)$-paths.  Moreover,
we prove a primal-dual result showing that, for any commodity graph $R$ on $X$ and any $r\ge 8\delta,$ the
size $\sigma_r(R)$  of the least  $r$-multi-core (i.e., the number of balls of radius $r$) intercepting all pairs of $R$ is upper bounded by the maximum  number of pairwise $(2r-5\delta)$-apart pairs of $R$ and that an $r$-multi-core of size $\sigma_{r-5\delta}(R)$ can be computed in polynomial time.

Our result about total $r$-multi-cores is based on a Helly-type theorem for quasiconvex sets in $\delta$-hyperbolic graphs (this is our second main result).
Namely,  we show that  for any finite collection $\mathcal Q$ of pairwise intersecting $\epsilon$-quasiconvex sets of a $\delta$-hyperbolic graph $G$ there exists a single ball $B(c,2\epsilon+5\delta)$ intersecting all sets of $\mathcal Q$. More generally, we prove that if $\mathcal Q$ is a collection of $2r$-close (i.e., any two sets of $\mathcal Q$ are at distance $\le 2r$) $\epsilon$-quasiconvex sets of a $\delta$-hyperbolic graph $G$, then there exists a ball $B(c,r^*)$ of radius $r^*:=\max\{ 2\epsilon+5\delta, r+\epsilon+3\delta\}$ intersecting all sets of $\mathcal Q$. These kind of Helly-type results are also useful in geometric group theory.

Using the Helly theorem for quasiconvex sets and  a primal-dual approach, we show algorithmically that the minimum number of balls of radius $2\epsilon+5\delta$ intersecting all sets of a family $\mathcal Q$ of $\epsilon$-quasiconvex sets does not exceed the packing number of $\mathcal Q$ (maximum number of pairwise disjoint sets of $\mathcal Q$).  We extend the covering and packing result to set-families $\kb$ in which each set  is a union of at most $\kappa$ $\epsilon$-quasiconvex sets of a $\delta$-hyperbolic graph $G$. Namely, we show that if $r\ge \epsilon+2\delta$ and $\pi_r(\kb)$ is the maximum number of mutually $2r$-apart members of $\kb$, then the minimum number of balls of radius $r+2\epsilon+6\delta$ intersecting all members of $\kb$ is at most $2\kappa^2\pi_r(\kb)$ and such a hitting set  and a packing can be constructed in polynomial time (this is our third main result). For set-families consisting of unions of $\kappa$ balls in $\delta$-hyperbolic graphs a similar result was obtained by Chepoi and Estellon (2007). In case of $\delta=0$ (trees) and $\epsilon=r=0,$ (subtrees of a tree) we recover the result of Alon (2002) about the transversal and packing numbers of a set-family in which each set is a union of at most $\kappa$ subtrees of a tree.
}


\section{Introduction}
Understanding key structural properties of large-scale data networks is crucial for analyzing and
optimizing their performance, as well as improving their reliability and security. 
In prior empirical  and theoretical studies researchers have mainly focused on features such as  small world phenomenon, power law degree distribution, navigability, and high clustering coefficients (see
\cite{barabasi99emergence,barabasi2000scalefree,Boguna2009,DBLP:journals/im/ChungL03,DBLP:conf/sigcomm/FaloutsosFF99,DBLP:conf/stoc/Kleinberg00,DBLP:conf/nips/Kleinberg01,DBLP:journals/im/LeskovecLDM09,Watts-Colective-1998}). Those nice features
were observed 
in many real-world complex networks and their underlying
graphs arising in Internet applications, in biological and social sciences, and in chemistry and physics. Although those features are interesting and important, as noted in \cite{NaSa}, the impact of intrinsic geometric and topological features of large-scale data networks on performance, reliability and security is of much greater importance.

Recently, there has been a surge of empirical works measuring and analyzing geometric characteristics of real-world networks, namely the
{\em hyperbo\-licity} (sometimes called also the {\em negative curvature}) of the network (see, e.g., \cite{{MAA-Dr-16},DBLP:conf/icdm/AdcockSM13,conf/isaac/ChenFHM12,Kennedy2013Arch,conf/nca/MontgolfierSV11,NaSa,DBLP:journals/ton/ShavittT08}). It has been shown that a number of data networks, including Internet application networks, web networks, collaboration networks, social networks, and others, have small hyperbolicity.
It has been suggested (see \cite{JoLoBoBa,NaSa}) that the property, observed in real-world networks,  in which traffic between vertices (nodes) tends to go through a
relatively small core of the network, as if the shortest path between them is curved inwards, may be due to global curvature
of the network.

In this paper, we prove that any finite subset $X$ of vertices in a locally finite $\delta$-hyperbolic graph $G$ admits a core, namely there exists a vertex $m$ of $G$ such that the ball $B(m,4 \delta)$ centered at $m$ of radius $4 \delta$ intersects all geodesics (shortest paths) between at least one half of all pairs of vertices of $X$. This solves in the positive and in the stronger  form the first part of Conjecture 1 of \cite{JoLoBoBa}, asserting: {\it ``Consider a large but finite negatively curved graph $G$, subject  to the uniformly distributed demand. Then there are very few nodes $v$ that have very high traffic rate...''}. This phenomenon was observed experimentally in \cite{NaSa} in some real-world networks with small hyperbolicity. On the other hand, we show that the vertex $m$  is not a center of mass as conjectured in \cite{JoLoBoBa} ({\it ``...furthermore, the vertices of highest traffic rate are in a small neighborhood of the vertices of minimum inertia''}) but is a vertex of $G$ close to a median point of $X$ in the injective hull of $G$. This confirms the experimental observation of \cite{NaSa} that {\it ``... the core is close to the geometric center, defined as the node whose average (geodesic) distance to all other nodes in the graph is the smallest.''} Notice also that the authors of \cite{JoLoBoBa} established their conjecture for a particular case of graphs that are quasi-isometric to the balls  of the $n$-dimensional hyperbolic space ${\mathbb H}^n$.

We also consider the case of non-uniform traffic between vertices of $X$. In this case, a unit demand of flow exists only between certain pairs of vertices of $X$ defined by a
commodity graph $R$; as in the previous case, the traffic between any pair of vertices defining an edge of $R$ is evenly distributed
over all geodesics connecting them. We prove a primal-dual result showing that for any $r\ge 8\delta$ the size of an $r$-multi-core (i.e., the number of balls of radius $r$) intercepting all pairs of $R$  is
upper bounded by the maximum  number of pairwise $(2r-5\delta)$-apart pairs of $R$. Finally, if $R$ consists of all mutually distant vertex pairs of a finite
$\delta$-hyperbolic graph $G$, then a single ball $B(m,2\delta)$ of radius $2\delta$ intercepts all pairs of $R$.

The proofs of all our results about cores implicitly or explicitly use various Helly type properties for balls, geodesics,  and intervals in $\delta$-hyperbolic graphs.  For example, the proof of our main result about existence of cores is based on the fact that, for any metric space $(X,d),$ there exists the smallest hyperconvex space $E(X)$ (i.e., geodesic metric space in which balls satisfy the Helly property) into which $(X,d)$ isometrically embeds; $E(X)$ is called the injective hull of $X$ \cite{Dr,Is}. We use a result of Lang \cite{La} asserting that if $(X,d)$ is $\delta$-hyperbolic, then $E(X)$ is also $\delta$-hyperbolic and if, in addition, $X$ is geodesic or a graph, then any point of $E(X)$ is within distance $\delta$ from some point of $X$. This last result is also a consequence of the Helly property for balls establishes in \cite{ChEs}.

The second main result of our paper is a general Helly-type theorem for quasiconvex sets in $\delta$-hyperbolic graphs, extending similar results for balls, geodesics, and intervals.
Namely,  we show that  for any finite collection $\mathcal Q$ of pairwise intersecting $\epsilon$-quasiconvex sets of a $\delta$-hyperbolic graph $G$ there exists a single ball $B^*$ of radius $2\epsilon+5\delta$
intersecting all sets of $\mathcal Q$. More generally, we prove that if $\mathcal Q$ is a collection of $2r$-close (i.e., any two sets of $\mathcal Q$ are at distance $\le 2r$) $\epsilon$-quasiconvex sets of a $\delta$-hyperbolic graph $G$, then there exists a ball $B^*$ of radius $r^*:=\max\{ 2\epsilon+5\delta, r+\epsilon+3\delta\}$ intersecting all sets of $\mathcal Q$. Niblo and Reeves \cite[Lemma 7]{NiRe} and implicitly Sageev \cite{Sa_co}   established this kind of Helly-type property for $\epsilon$-quasiconvex sets in $\delta$-hyperbolic graphs (see also \cite[Proposition 7.7]{HrWi} for a generalization to relatively hyperbolic groups), but in their result the radius of the ball $B^*$ hitting the sets of $\mathcal Q$ depends also on the number of sets in $\mathcal Q$. This statement plays a fundamental role in the cubulation process in proving the cocompactness of the cube complex associated with a finite set of quasiconvex codimension-1 subgroups \cite{HrWi,NiRe,Sa_co}. The Helly property for balls proved in \cite{ChEs} is also important in the dismantlability and cop-and-robber game characterizations of hyperbolic graphs established in \cite{ChChPaPe}.

 Using the Helly theorem for quasiconvex sets and  a primal-dual approach, we show algorithmically that the minimum number of balls of radius $2\epsilon+5\delta$ intersecting all sets of a family $\mathcal Q$ of $\epsilon$-quasiconvex sets does not exceed the packing number of $\mathcal Q$ (maximum number of pairwise disjoint sets of $\mathcal Q$). The Helly property for geodesics and intervals is used to establish the existence of total beam cores and the covering and packing result is used in the computation of total multi-cores. Then we extend the covering and packing result from set-families  $\mathcal Q$ consisting of quasiconvex sets  to set-families $\kb$ in which each set  is a union of at most $\kappa$ $\epsilon$-quasiconvex sets of a $\delta$-hyperbolic graph $G$.  Namely, we show that if $r\ge \epsilon+2\delta$ and $\pi_r(\kb)$ is the maximum number of mutually $2r$-apart members of $\kb$, then the minimum number of balls of radius $r+2\epsilon+6\delta$ intersecting all members of $\kb$ is at most $2\kappa^2\pi_r(\kb)$ and such a hitting set  and a packing can be constructed in polynomial time (this is our third main result). For set-families consisting of unions of $\kappa$ balls in $\delta$-hyperbolic graphs a similar result was obtained in \cite{ChEs} (and we closely follow the local-ratio proof-techniques of \cite{ChEs} and \cite{BYHaNaShSh}). In case of $\delta=0$ (trees) and $\epsilon=0$ (subtrees of a tree) we recover the result of Alon \cite{Al} about the transversal and packing numbers of a set-family in which each set is a union of at most $\kappa$ subtrees of a tree (for intervals of a line a similar inequality was proved in \cite{Al1,BYHaNaShSh}). Thus our result can be viewed as a far-reaching generalization of the result of \cite{Al} in which trees are replaced by hyperbolic graphs and subtrees by quasiconvex subgraphs.

%

\section{Preliminaries}

\subsection{Graphs}

All graphs $G=(V,E)$ occurring in this paper are undirected,
connected, without loops or multiple edges, but not necessarily finite. 
For a subset $A\subseteq V,$ the subgraph of $G=(V,E)$  {\it induced by} $A$
is the graph $G(A)=(A,E')$ such that $uv\in E'$ if and only if $u,v\in A$ and $uv\in E$.
The {\it distance} $d(u,v):=d_G(u,v)$
between two vertices $u$ and $v$ of  $G$ is the length (number of
edges) of a $(u,v)$-{\it geodesic}, i.e., a shortest $(u,v)$-path. For a vertex $v$ of $G$ and an
integer $r\ge 0$, we will denote  by $B(v,r)$ the \emph{ball} in $G$
of radius $r$ centered at  $v$, i.e.,
$B(v,r)=\{ x\in V: d(v,x)\le r\}.$  The {\it interval}
$I(u,v)$ between $u$ and $v$ consists of all vertices on
$(u,v)$-geodesics, that is, of all vertices (metrically) {\it between} $u$
and $v$:
$$I(u,v)=\{ x\in V: d(u,x)+d(x,v)=d(u,v)\}.$$
Let $d(X,Y)=\min \{ d(x,y): x\in X,y\in Y\}$ denote the distance between two subsets $X,Y$ of vertices of $G$. We will say that two sets $X$ and $Y$
are $r$-{\it close} if $d(X,Y)\le r$ and that  $X$ and $Y$ are $r$-{\it apart} if $d(X,Y)>r$. In particular, two
intersecting sets are $0$-close.

We will call any finite subset $X$ of vertices of a graph $G$ a {\it profile}. Given a profile $X$,
any vertex $v$ of $G$ minimizing the distance sum $\Psi_X(v):=\sum_{x\in X} d(v,x)$ is called a {\it median vertex} of $X$.
Analogously, any vertex $v$ of $G$ minimizing the sum $\Phi_X(v):=\sum_{x\in X} d^2(v,x)$ is called a {\it center of mass}
or a {\it centroid} of $X$.

Given a finite set $X$ of vertices of a graph $G$, the
{\em diameter} $\diam(X)$ of $X$ is the maximum distance between any two vertices of $X$.
A {\em diametral pair} of $X$ is any pair of vertices $x, y\in X$ such that $d(x, y) = \diam(X)$.
For a vertex $x$ of a graph $G$ of finite diameter, the set $P(x)$ of {\em furthest neighbors} of
$x$ (or of {\it peripheral  with respect} to $x$ vertices)
consists of all vertices of $G$ located at the maximum distance
from $x$.  The {\em eccentricity} $\ecc(x)$ of a vertex  $x$ is the distance from $x$ to any vertex of $P(x)$.
The {\em center} $C(G)$ of $G$ is the set of all vertices of $G$ having minimum
eccentricity; the vertices of $C(G)$ are called central vertices. The radius $\rad(G)$ of $G$ is the eccentricity of
its central vertices. A geodesic $[x,y]$ between two vertices $x,y$ such that $y\in P(x)$ is called a {\em beam} and $\{ x,y\}$ is called a
{\it beam pair} of $G$). Two vertices $x,y$ of a graph $G$ are called
{\em mutually distant} if $x\in P(y)$ and $y\in P(x)$. 

Given  a graph $G=(V,E)$, a subset $X$ of vertices of $G$, and a set of pairs $F\subseteq X\times X$, analogously
to the multicommodity flow  terminology (see also the next subsection), the pair $R=(X,F)$ will be called a
{\it commodity graph}. 

\subsection{Cores}

We say that a ball $B(v,r)$ {\it intercepts} a geodesic $[x,y]$ of $G$ if $B(v,r)\cap [x,y]\neq \varnothing$. More generally, we will say that a
ball $B(v,r)$ {\it intercepts} a pair of vertices $x,y$ if $B(v,r)$ intercepts all geodesics $[x,y]$ between $x$ and $y$. 

Given $0<\alpha\le 1$ and $r\ge 0$, we will say that a graph $G$ has an $(\alpha, r)$-{\it core} if for any profile $X$ in $G$
there exists a vertex $m:=m(X)$ such that the ball $B(m,r)$ intercepts strictly more than the fraction of $\alpha$ of all
pairs of $X$,  i.e., there exist more than $\frac{\alpha|X|(|X|-1)}{2}$ pairs $\{ x,y\}$ of $X$ intercepted by $B(m,r)$.

Given an integer $r\ge 0$, we will say that a graph $G$ admits a {\it total beam $r$-core} if
there exists a ball $B(v,r)$ of radius $r$ intercepting all beam pairs of $G$. More generally, given a graph $G=(V,E)$, an integer $r\ge 0$,
and a commodity graph $R=(X,F)$ with a profile $X\subseteq V$,
we will say  that $R=(X,F)$ has a {\it total $r$-multi-core of size $k$} if all pairs of $F$ can be intercepted with $k$ balls of radius $r$.
This last definition of multi-core corresponds to the model in which the traffic is not uniform but is performed only among the pairs  of vertices
of $X$ defined by the commodity graph $R$. We will denote by $\sigma_r(R)$ the least integer $k$ such that the commodity graph $R=(X,F)$ has a total
$r$-multi-core of size $k$.

\subsection{Traffic metrics and cores}

Following \cite{JoLoBoBa}, let us consider a network in which the traffic is driven by a demand measure $\Lambda_d:V\times V \leftarrow \mathbb{R}^+,$ where the demand $\Lambda_d(s,t)$ is the traffic rate (e.g. the number of packets per second) to be transmitted from the source $s$ to the destination target $t.$ Assume that the routing protocol sends packets from source $s$ to target $t$ along the geodesic $[s,t]$ with probability $Pr([s,t]).$ It is 
customary as a load balancing strategy to randomize the Dijkstra algorithm so as to distribute the traffic more evenly. Under this scheme, the geodesic $[s,t]$ inherits a traffic  rate measure $\mu([s,t]) := \Lambda_d(s,t) Pr([s,t]).$ A subset of $S$ of vertices crossed by a path $[s,t]$ inherits from that path a traffic $\mu([s,t]).$ Aggregating this traffic over all source-target pairs and all geodesics traversing $S$, yields the traffic rate sustained by the subset $S$:

$$\mu(S) := \sum_{(s,t)\in V\times V}\sum_{[s,t]\cap S \neq \varnothing} \mu([s,t]).$$

In this paper, we will consider both uniform and non uniform traffic. In case of uniform traffic,
we show that in any $\delta$-hyperbolic network $G$, there exists a ball of radius $O(\delta)$ that has an  extremely high traffic load in the sense that the majority of the traffic passes through this ball. In case of non uniform traffic, we consider a family of geodesics on which the traffic is sent and show that the minimum number of balls of radius $O(\delta)$ needed to collectively intercepts all these geodesics is bounded by the maximum number of pairwise $O(\delta)$-apart geodesics in this family. 

\subsection{Hitting and packing problems}

The hitting and packing problems are classical problems in computer
science and combinatorics. Let $\mathcal S$ be  a finite  collection
of subsets of a domain $V$. A subset $T$ of $V$ is called a
{\it hitting set}  of $\mathcal S$ if $T\cap S\ne\varnothing$
for any $S\in {\mathcal S}.$ The {\it
minimum hitting set problem} asks to find a hitting set of $\mathcal
S$ of smallest cardinality $\tau ({\mathcal S})$. The {\it set packing problem}
(dual to the hitting set problem) asks to find a maximum
number $\pi({\mathcal S})$ of pairwise disjoint subsets of $\mathcal
S$. We will call $\tau ({\mathcal S})$ and $\pi({\mathcal S})$
the {\it transversal}
  (or  {\it hitting}) and the {\it packing numbers}
of $\mathcal S$. Obviously, the inequality $\tau ({\mathcal S})\ge \pi({\mathcal S})$ holds for any
set-family $\mathcal S$.

In this paper, the domain $V$ is the set of vertices of a
connected graph $G=(V,E)$ or the set of points of a metric space $(V,d)$. 
It this case, we can formulate the following relaxed hitting set problem. 
For $r\ge 0$, the $r$-{\it neighborhood} of $S$ is the set
$N_{r}(S):=\bigcup_{v\in S} B_r(v)$. For a collection of sets $\mathcal S$, let
${\mathcal S}_{r}:=\{ N_{r}(S): S\in {\mathcal S}\}$; we will sometime refer to
${\mathcal S}_{r}$ as to the $r$-{\it inflation} of the collection $\mathcal S$.
For $r\ge 0$, a subset $T$ of $V$ is called an
{\it $r$-hitting set}  of $\mathcal S$ if for any $S\in {\mathcal S}$ there exists $t\in T$
such that $B(t,r)\cap S\ne \varnothing$.
The {\it minimum $r$-hitting set problem} asks to find an $r$-hitting set of $\mathcal
S$ of smallest cardinality $\tau_{r}({\mathcal S}).$ Notice that
$\tau_{r}({\mathcal S})=\tau({\mathcal S}_{r}).$  Analogously, a subfamily ${\mathcal P}$ of $\mathcal S$ is called an $r$-{\it packing}
if $N_r(S)\cap N_r(S') = \varnothing$ for any $S,S'\in {\mathcal P}$, i.e.,  if any two sets of $\mathcal P$ are $2r$-apart.  We will be interested in set-families
$\mathcal S$ such that for any $r\ge 0$ and for some constant $\alpha$ not depending on the family $\mathcal S$,
$\tau({\mathcal S}_{r+\alpha})$ is upper bounded  by $\pi({\mathcal S}_{r})$. In our case, if $\mathcal S$ is
a collection of $\epsilon$-quasiconvex sets of a $\delta$-hyperbolic graph $G$, then $\alpha$ will be a constant
depending only on $\epsilon$ and $\delta$.

\section{Gromov hyperbolicity}

\subsection{Definition, characterizations, and properties}
Let $(X,d)$ be a metric space and $w\in X$. The {\it Gromov product} of $y,z\in X$ with respect to $w$ is defined to be
$$(y|z)_w=\frac{1}{2}(d(y,w)+d(z,w)-d(y,z)).$$
Let $\delta\ge 0$. A metric space $(X,d)$ is said to be $\delta$-{\it hyperbolic} \cite{Gr} if
$$(x|y)_w\ge \min \{ (x|z)_w, (y|z)_w\}-\delta$$
for all $w,x,y,z\in X$. Equivalently, $(X,d)$ is $\delta$-hyperbolic
if  for any four points $u,v,x,y$ of $X$, the two larger of the three distance sums
$d(u,v)+d(x,y)$, $d(u,x)+d(v,y)$, $d(u,y)+d(v,x)$ differ by at most
$2\delta \geq 0$. In case of geodesic metric spaces and graphs, there exist several equivalent
definitions of $\delta$-hyperbolicity involving different but
comparable values of $\delta$ \cite{AlBrCoFeLuMiShSh,BrHa,GhHa,Gr}.

\commentout{\FFD{(THIS SECTION IS TAKEN FROM EARLIER PAPERS: \\
WE MAY NEED TO ADAPT EVERYTHING TO GRAPHS ONLY; \\
IF WE START WITH $\delta$-thin triangles in graphs \\
then the corresponding geodesic space will be $(\delta+1)$-thin;\\
IF WE START WITH $\delta$-thin triangles in geodesic spaces \\
then the graph must be $(\delta-1)$-thin to obtain from it a geodesic space with $\delta$-thin triangles.)}\\
}

Let $(X,d)$ be a metric space.  A {\it geodesic segment} joining two
points $x$ and $y$ from $X$ is a (continuous) map $\rho$ from the segment $[a,b]$
of ${\mathbb R}^1$ of length $|a-b|=d(x,y)$ to $X$ such that
$\rho(a)=x, \rho(b)=y,$ and $d(\rho(s),\rho(t))=|s-t|$ for all $s,t\in
[a,b].$ A metric space $(X,d)$ is {\it geodesic} if every pair of
points in $X$ can be joined by a geodesic segment. Every (combinatorial)
graph $G=(V,E)$ equipped with its standard
distance $d:=d_G$ can be transformed into a geodesic (network-like)
space $(X_G,d)$ by replacing every edge $e=(u,v)$ by a segment
$\gamma_{uv}=[u,v]$ of length 1; the segments may intersect only at
common ends.  Then $(V,d_G)$ is isometrically embedded in a natural
way in $(X_G,d)$. $X_G$ is often called a {\it metric graph}.  The
restrictions of geodesics of $X_G$ to the set of vertices $V$ of $G$
are the shortest paths of $G$.
For simplicity of notation and brevity (and if not said otherwise), in all
subsequent results, by a geodesic $[x,y]$ in a graph $G$ we will mean an arbitrary shortest
path between two vertices $x,y$ of $G$.

Let $(X,d)$ be a geodesic metric space.  A \textit{geodesic triangle}
$\Delta(x,y,z)$ with $x, y, z \in X$ is the union $[x,y] \cup [x,z]
\cup [y,z]$ of three geodesic segments connecting these vertices.  A
geodesic triangle $\Delta(x,y,z)$ is called $\delta$-{\it slim} if for
any point $u$ on the side $[x,y]$ the distance from $u$ to $[x,z]\cup
[z,y]$ is at most $\delta$. Let $m_x$ be the point of the
geodesic segment $[y,z]$ located at distance $\alpha_y :=(x|z)_y=
(d(y,x)+d(y,z)-d(x,z))/2$ from $y.$ Then $m_x$ is located at
distance $\alpha_z :=(y|x)_z= (d(z,y)+d(z,x)-d(y,x))/2$ from $z$ because
$\alpha_y + \alpha_z = d(y,z)$. Analogously, define the points
$m_y\in [x,z]$ and $m_z\in [x,y]$ both located at distance $\alpha_x
:=(y|z)_x= (d(x,y)+d(x,z)-d(y,z))/2$ from $x;$ see Fig.~\ref{fig1} for an
illustration. There exists a unique isometry $\varphi$ which maps
$\Delta(x,y,z)$ to a star $\Upsilon(x',y',z')$ consisting of three
solid segments $[x',m'],[y',m'],$ and $[z',m']$ of lengths
$\alpha_x,\alpha_y,$ and $\alpha_z,$ respectively. This isometry
maps the vertices $x,y,z$ of $\Delta(x,y,z)$ to the respective
leaves $x',y',z'$ of $\Upsilon(x',y',z')$ and the points $m_x,m_y,$
and $m_z$ to the center $m$ of this tripod. Any other point  of
$\Upsilon(x',y',z')$ is the image of exactly two points of $\Delta
(x,y,z).$ A geodesic triangle $\Delta(x,y,z)$ is called
$\delta$-{\it thin} if for all points $u,v\in \Delta(x,y,z),$
$\varphi(u)=\varphi(v)$ implies $d(u,v)\le \delta.$
The notions of geodesic triangles,
$\delta$-slim and $\delta$-thin triangles can be also defined in
case of graphs. The single difference is that for graphs, the center
of the tripod is not necessarily the image of any vertex on the
geodesic of $\Delta(x,y,z).$ Nevertheless, if a point of the tripod
is the image of a vertex of one side of $\Delta(x,y,z),$ then it is
also the image of another vertex located on another side of
$\Delta(x,y,z).$
A {\em graph with $\delta$-thin triangles} is a graph  $G$ where each geodesic triangle is $\delta$-thin.

The following results show that hyperbolicity of a
geodesic space is equivalent to having thin or slim geodesic
triangles (the same result holds for graphs).

\begin{proposition} \cite{AlBrCoFeLuMiShSh,BrHa,Gr,GhHa}\label{hyp_charact}
Geodesic triangles of geodesic $\delta$-hyperbolic spaces are $4\delta$-thin
 and $3\delta$-slim.
\end{proposition}

We will use the following converse given in \cite[p. 411, Proposition 1.22]{BrHa} (since we often use the fact that $\delta$-thin triangles imply the $\delta$-hyperbolicity,
we will present a proof for the completeness):

\begin{lemma}\label{hyp_thin} A geodesic space $(X,d)$ or a graph with
$\delta$-thin triangles is $\delta$-hyperbolic.
\end{lemma}

\begin{proof} We will prove that for any four points $w,x,y,z$, we have
$(x|y)_w\ge \min \{ (x|z)_w, (y|z)_w\}-\delta$. Consider two geodesic triangles $\Delta(x,y,w)$ and $\Delta(y,z,w)$ sharing the common geodesic
$[y,w]$. Suppose without loss of generality that $\alpha:=(y|z)_w\le (x|y)_w$. Let $x'$ and $y'$ be two points on the geodesics $[w,x]$ and $[w,y]$, respectively,
located at distance $(x|y)_w$ from $w$. Analogously, let $z'$ and $y''$ be  two points on the geodesics $[w,z]$ and $[w,y]$, respectively,
located at distance $\alpha$ from $w$. Since $\alpha \le (x|y)_w$, the point $y''$ is located on the geodesic $[w,y]$ between $w$ and $y'$. Let $x''$ be a point
of $[w,x]$ located at distance $\alpha$ from $w$. Again, $x''$ is located on $[w,x]$ between $w$ and $x'$. From the definition of the points $x'',y'',z'$ we conclude that
$d(x'',y'')\le \delta$ and $d(y'',z')\le \delta$. Hence, by the triangle inequality, we obtain
$$d(x,z)\le d(x,x'')+d(x'',y'')+d(y'',z')+d(z',z)=(d(x,w)-\alpha)+2\delta+(d(z,w)-\alpha).$$
By definition, $(x|z)_w=\frac{1}{2}(d(x,w)+d(z,w)-d(x,z))$. Replacing in the right-hand side the previous inequality for $d(x,z)$, we obtain that
$$(x|z)_w\ge \frac{1}{2}(d(x,w)+d(z,w)-d(x,w)-d(z,w)+2\alpha-2\delta),$$
whence
$$(x|z)_w\ge \frac{2\alpha-2\delta}{2}=\min \{ (x|z)_w, (y|z)_w\}-\delta.$$
\end{proof}



\begin{figure}
\begin{center}
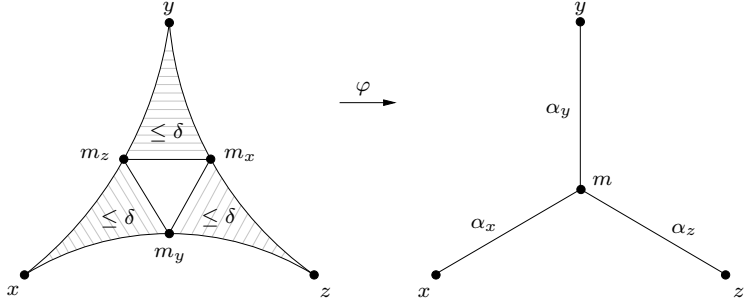
\end{center}
\caption{A geodesic triangle $\Delta(x,y,z),$ the points $m_x, m_y,
m_z,$ and the tripod $\Upsilon(x',y',z')$} \label{fig1}
\end{figure}


An interval $I(u,v)$ of a graph (or a geodesic metric space) is called $\nu$-{\it thin}, if $d(x,y)\le \nu$ for any two points $x,y\in I(u,v)$ such that $d(u,x)=d(u,y)$ and $d(v,x)=d(v,y).$
From the definition of $\delta$-hyperbolicity easily follows that intervals of $\delta$-hyperbolic geodesic metric spaces or graphs are $2\delta$-thin. In case of graphs (or geodesic spaces)
with $\delta$-thin triangles,  a better bound holds:

\begin{lemma}\label{interval_thin} Intervals of a  graph $G$  (or geodesic space) with $\delta$-thin geodesic triangles  are $\delta$-thin.
\end{lemma}

\begin{proof} Let $u,v$ be two arbitrary vertices of $G$ and let $x,y\in I(u,v)$ such that $d(u,x)=d(u,y)$. Let $[u,v]$ be any $(u,v)$-geodesic passing via $x$ and $[u,x],[x,v]$ be two arbitrary
$(u,x)$- and $(x,v)$-geodesics. Consider the geodesic triangle $\Delta(x,u,v):=[u,x]\cup [x,v]\cup [v,u]$ and define the points $v'\in [u,x], u'\in [x,v],$ and $x'\in [u,v]$ such that
$d(u,v')=d(u,x')=(x|v)_u, d(x,v')=d(x,u')=(u|v)_x$, and $d(v,u')=d(v,x')=(u|x)_v$. Since $d(u,x)+d(x,v)=d(u,v)=d(u,x')+d(x',v)$, necessarily  $v'=x=u'$. Consequently,
$d(u,x)=d(u,x')$ and $d(v,x)=d(v,x')$, i.e., $x'=y$. Since $\Delta(x,u,v)$ is $\delta$-thin, $d(v',x')\le \delta$, yielding $d(x,y)\le \delta$.
\end{proof}

By this lemma, any result about cores intercepting families of geodesics can be transformed into a result about cores intercepting all pairs of vertices corresponding to ends of those
geodesics.


\subsection{Quasiconvexity}
A subset $C$ of a geodesic metric space or graph is called {\it convex} if for all $x,y\in C$ each geodesic joining $x$ and $y$ is contained in $C$. The following ``quasification'' of this notion
due to Gromov \cite{Gr} plays an important role in the study of hyperbolic and cubical groups \cite{BrHa,HrWi,NiRe,Sa_co}. For $\epsilon\ge 0$, a subset $C$ of a geodesic metric space $(X,d)$ or graph $G=(V,E)$ is called $\epsilon$-{\it quasiconvex} if for all $x,y\in C$ each geodesic joining $x$ and $y$ is contained in the $\epsilon$-neighborhood $N_{\epsilon}(C)$ of $C$. $C$ is said to be {\it quasiconvex} if there exists a constant $\epsilon\ge 0$ such that $C$ is $\epsilon$-quasiconvex. It turns out that in $\delta$-hyperbolic spaces the collection of quasiconvex sets is abundant and it contains, in particular, geodesics, intervals, and balls:

\begin{lemma} \label{quasiconvex} Let  $G$ be a graph (or geodesic space)  with $\delta$-thin geodesic triangles. Then the geodesics, the intervals, and the balls of $G$ are $\delta$-quasiconvex, and the neighborhoods of  $\epsilon$-quasiconvex sets are $(\epsilon+2\delta)$-quasiconvex.
\end{lemma}

\begin{proof} That geodesics are $\delta$-quasiconvex immediately follows from the fact that the intervals are $\delta$-thin (Lemma \ref{interval_thin}). To prove that any interval $I(u,v)$ is $\delta$-quasiconvex, pick any two points $x,y\in I(u,v)$ and any geodesic $[x,y]$ between $x$ and $y$. Let $[u,x]$ and $[u,y]$ be two arbitrary geodesics between $u$ and $x$ and $u$ and $y$, respectively. Since the resulting geodesic triangle $\Delta(u,x,y)$ is $\delta$-thin, any point of $[x,y]$ is at distance at most $\delta$ from a point of $[u,x]$ or $[u,y]$. Since $[u,x]\cup [u,y]\subset I(u,v)$, $[x,y]$ is contained in the $\delta$-neighborhood of $I(u,v)$ and we are done. The proof that balls are $\delta$-quasiconvex is analogous.

Finally suppose that $C$ is an $\epsilon$-quasiconvex set of $(X,d)$ and for $r\ge 0$ let $N_r(C)$ be the $r$-neighborhood of $C$. Let $x,y\in N_r(C)$ and $x',y'\in C$ such that $d(x,x'),d(y,y')\le r$. Pick any geodesics $[x,y],[x,y'],[x,x'],[x',y'],$ and $[y',y]$. Notice that $[x,x']\cup [y,y']\subset N_r(C)$. Let $\Delta(x,y,y')$ be the geodesic triangle with sides $[x,y],[x,y'],[y',y]$ and $\Delta(x,x',y')$ be the geodesic triangle with sides $[x,x'],[x',y'],[x,y']$. Let $z$ be any point of $[x,y]$. Since  $\Delta(x,y,y')$ is $\delta$-thin, $z$ is at distance at most $\delta$ from some point $z'\in [x,y']\cup [y,y']$. If $z'\in [y,y']\subset N_r(C)$, then $d(z,z')\le \delta$ and we are done. So, suppose that $z'\in [x,y']$. Since $\Delta(x,x',y')$ is $\delta$-thin, $z'$ is at distance at most $\delta$ from a point $z''\in [x,x']\cup [x',y']$. Again, if $z''\in [x,x']\subset N_r(C)$, then $d(z,z'')\le d(z,z')+d(z',z')\le 2\delta$ and we are done. Finally, if $z''\in [x',y']$, since $C$ is $\epsilon$-quasiconvex, there exists a point $p\in C$ such that $d(z'',p)\le \epsilon$. Consequently, $d(z,p)\le d(z,z'')+d(z'',p)\le \epsilon+2\delta$.
\end{proof}

\subsection{Injective hulls of Gromov hyperbolic spaces}
A metric space  $(X,d)$ is said to be {\it injective} if, whenever $X$ is isometric to a subspace $Z$ of a metric space $(Y,d')$, then the subspace $Z$ is a retract of $Y$, i.e., there exists a map
$f:Y\rightarrow Z$ such that $f(z)=z$ for any $z\in Z$ and $d'(f(x),f(y))\le d'(x,y)$ for any $x,y\in X$. As shown in \cite{ArPa}, injectivity of a metric space $(X,d)$ is equivalent to its hyperconvexity. A metric space  $(X,d)$ is said to be {\it hyperconvex} if it is a geodesic metric space and its closed balls satisfy the {\it Helly property},  i.e.,    if $\mathcal B$  is any family of closed balls  of $X$ such that each pair of balls in $\mathcal B$ meet, then there exists a point $x$ common to all the balls in $\mathcal B$. By a construction of Isbell \cite{Is}, for any metric space $(X,d)$ there exists an essentially unique {\it injective hull} $(e,E(X))$, that is $E(X)$ is an injective metric space, $e: X\rightarrow E(X)$ is an isometric embedding, and every isometric embedding of $X$ into some injective metric space $Z$ implies an isometric embedding  of $E(X)$ into $Z$ (thus  $E(X)$ is the smallest injective space containing an image of an isometric embedding of $X$). This construction was rediscovered later by Dress \cite{Dr}. It was noticed without any proof in \cite{DrMouTe} that the injective hull of a $\delta$-hyperbolic space is $\delta$-hyperbolic.
This result was rediscovered recently by Lang \cite{La}, who also proved that if $(X,d)$ is a geodesic space and a graph, then any point of $E(X)$ is located at distance at most $2\delta$ (respectively, $2\delta+\frac{1}{2}$) from a point of $X$. Since we use this Lang's result, we briefly recall the basic definitions about injective hulls (in which we closely follow \cite{La}).

Let  $(X,d)$ be a metric space. Denote by ${\mathbb R}^X$ the vector space of all real valued functions on $X$, and define
$$\Delta(X):=\{ f\in {\mathbb R}^X: f(x)+f(y)\ge d(x,y) \mbox{ for all } x,y\in X\}.$$
Notice that if $f\in \Delta(X)$, then ${\mathcal B}(f)=\{ B(x,f(x)): x\in X\}$ is a family of pairwise intersecting balls.
For a point $z\in X$ define the distance function $d_z\in {\mathbb R}^X$ by setting $d_z(x)=d(x,z)$ for any $x\in X$. By the triangle inequality, each $d_z$ belongs to $\Delta(X)$. A function $f\in \Delta(X)$ is called {\it extremal} if it is a minimal element of the partially ordered set $(\Delta(X),\le)$, where $g\le f$ means $g(x)\le f(x)$ for all $x\in X$. Let
$$E(X)=\{ f\in \Delta(X): \mbox{ if } g\in \Delta(X) \mbox{ and } g\le f, \mbox{ then } g=f\}$$
denote the set of all extremal functions on $X$. Then the {\it injective hull} of $X$ is the set $E(X)$ equipped with the $l_{\infty}$-metric $||f-g||_{\infty}=\sup_{x\in X}|f(x)-g(x)|$. It can be easily seen that the map $e:X\rightarrow E(X)$ defined by $e(z)=d_z$ for any $z\in X$ is a  canonical isometric embedding of $X$ into $E(X)$. Moreover, it was shown in \cite{Dr,Is,La} that $E(X)$ is an injective (and thus hyperconvex) space and it is minimal in this sense.

A similar construction can be  done if $X$ is a graph $G$; there exists a smallest Helly
graph (i.e., a graph satisfying the Helly property for balls) comprising $G$ as an isometric subgraph \cite{JaMiPo,Pe}. In this case, instead of taking the set $E(X)$ of all extremal functions, one can take the subset $E_0(X)$ of $E(X)$ consisting  only of integer-valued extremal functions; endow it with the $l_{\infty}$-metric, and consider the graph $H(G)$ having $E_0(X)$ as the vertex-set and all pairs of vertices having  $l_{\infty}$-distance 1 as edges. This graph $H(G)$ is called the {\it Hellyfication} of $G$.

Returning to $\delta$-hyperbolic spaces and graphs, in what follows, we will use the following result of Lang \cite{La}:

\begin{proposition} \cite[Proposition 1.3]{La} \label{lang} If $(X,d)$ is a $\delta$-hyperbolic metric space, then its injective hull is
$\delta$-hyperbolic. If, in addition, $X$ is a geodesic space or a graph, then any point $x^*$ of $E(X)$ (respectively, $E_0(X)$) is within
distance $2\delta$ 
from some point (respectively, some vertex) $x$ of $X$.
\end{proposition}

For geodesic spaces or graphs with $\delta$-thin triangles, the second assertion of Proposition~\ref{lang} with ($\delta$ instead of $2\delta$)
also follows from the following {\it Helly property for balls}:

\begin{proposition} \cite[Corollary 2]{ChEs} \label{helly-balls} Let $X$ be a geodesic space or graph with $\delta$-thin triangles  and  let
$B(x_i,r_i)$ be a collection of pairwise intersecting balls of $G$. Then the balls $B(x_i,r_i+\delta)$ have a
nonempty intersection.
\end{proposition}

\begin{proposition} \label{lang1}  If  $X$ is a geodesic space or a graph with $\delta$-thin triangles, then any point $x^*$ of $E(X)$ is within
distance
$\delta$ from some point $x$ of $X$.
\end{proposition}

\begin{proof} Pick any point $f\in E(X)$. Since $f(x)+f(y)\ge d(x,y)$,  $\{ B(x,f(x)): x\in X\}$ is a collection of pairwise intersecting balls of $X$. By
Proposition~\ref{helly-balls}, there exists a point $z$ belonging to all balls $B(x,f(x)+\delta), x\in X$. Consider the extremal map  $d_z$, i.e., the point
of $E(X)$ corresponding to $z$. Recall, that $d_z(x)=d(x,z)$ for any $x\in X$. By definition of $z$, for any $x\in X$ we have $d_z(x)=d(x,z)\le f(x)+\delta$.
On the other hand, if there exists $x\in X$ such that $f(x)>d_z(x)+\delta$, then  we assert that $f$ is not an extremal map.  Indeed, for any $y\ne x$ of $X$, we will obtain
that $f(x)+f(y)>d_z(x)+\delta+f(y)\ge d_z(x)+d_z(y)\ge d(x,y)$, showing that $f$ is not extremal. Consequently,  $f(x)\le d_z(x)+\delta$ for any $x\in X$,
whence $||f-d_z||_{\infty}=\sup_{x\in X}|f(x)-d_z(x)|\le \delta$.
\end{proof}

\section{Existence of cores}

The goal of this section is to prove the following result:

\begin{theorem} [Existence of cores]\label{core} Let $G$ be a $\delta$-hyperbolic graph  (respectively, a graph with $\delta$-thin triangles). Then any finite subset $X$ of vertices of $G$
has a $(\frac{1}{2},4\delta)$-core (respectively, a $(\frac{1}{2},3\delta+\frac{1}{2})$-core). 
\end{theorem}

\begin{proof} Let $n:=|X|$. First suppose that $G$ is a graph with $\delta$-thin triangles. By Lemma~\ref{hyp_thin}, $G$ is $\delta$-hyperbolic.
Let $E(G)$ be the injective hull of $G$ and let $H(G)$ be the Hellification of $G$ (induced
by all points of $E(G)$ with integer coordinates). By the first part of Proposition~\ref{lang}, $E(G)$ is $\delta$-hyperbolic, thus $H(G)$ is also $\delta$-hyperbolic.
Let $m^*$ be a median vertex of the profile $X$ in the Helly graph $H(G)$. By Proposition~\ref{lang1}, $m^*$ is at distance at most $\delta$
from a vertex $m$ of $G$.

For a vertex $x\in X$, let
$$F_{m^*}(x)=\{ y\in X: (x|y)_{m^*}\ge \delta+1\}$$
and call $F_{m^*}(x)$ the {\it fiber} of $x$ with respect to $m^*$.

\begin{figure}
\begin{center}
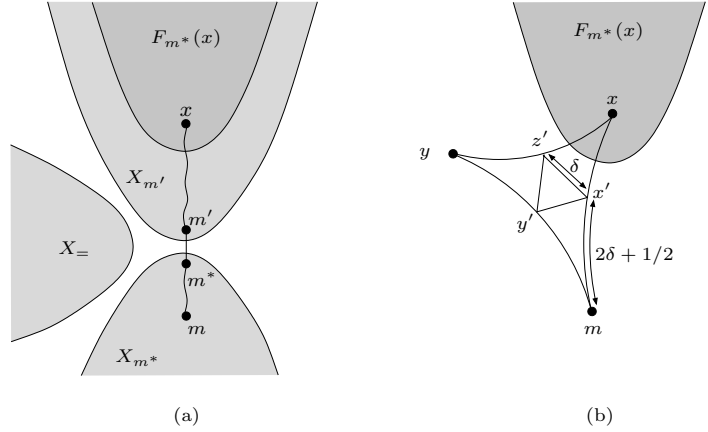
\end{center}
\caption{To the proof of Claim 1 (a) and Claim 3 (b).} \label{fig2}
\end{figure}

\medskip\noindent
{\bf Claim 1.} For any vertex $x\in X$, the fiber  $F_{m^*}(x)$ contains at most $n/2$ vertices.

\medskip\noindent
{\bf Proof of Claim 1.} Suppose, by way of contradiction, that $|F_{m^*}(x)|>n/2$. For each vertex $v\in F_{m^*}(x)$, set $r_v:=d(m^*,v)-1$.
Consider the following collection of balls:
$${\mathcal B}=\{ B(m^*,1)\}\cup \{ B(v,r_v): v\in F_{m^*}(x)\}.$$
We assert that  the balls from $\mathcal B$ pairwise intersect. From the definition of $r_v$, this is obviously true for $B(m^*,1)$ and $B(v,r_v)$ for any $v\in F_{m^*}(x)$.
Now, pick two arbitrary vertices $u,v\in F_{m^*}(x).$ We assert that $d(u,v)\le r_u+r_v$. Since $H(G)$ is $\delta$-hyperbolic and $u,v\in F_{m^*}(x),$
$$(u|v)_{m^*}\ge \min \{ (x|u)_{m^*},(x|v)_{m^*}\}-\delta\ge \delta+1-\delta=1.$$
Hence, $d(u,m^*)+d(v,m^*)-d(u,v)=2(u|v)_{m^*}\ge 2$. Consequently, $d(u,v)\le d(u,m^*)+d(v,m^*)-2=r_u+r_v$, showing that the balls $B(u,r_u)$ and $B(v,r_v)$ intersect.

Applying the Helly property to the collection ${\mathcal B}$, we can find a vertex $m'$ of $H(G)$ belonging to all balls of ${\mathcal B}$. Since $r_v=d(m^*,v)-1$ for any $v\in F_{m^*}(x)$, $m'$ is different from $m^*$. Consequently, $m'$ is a neighbor of $m^*$ belonging to all intervals $I(m^*,v), v\in F_{m^*}(x)$; see Figure~\ref{fig2}(a). Consider now the values of the median function $\Psi_X$ on the vertices $m^*$ and $m'$ of $H(G)$. Consider a partition of the set $X$ into three sets $X_{m^*},X_{m'},$ and $X_{=}$, where
$X_{m^*}:=\{ v\in X: d(m^*,v)<d(m',v)\},$ $X_{m'}:=\{ v\in X: d(m',v)<d(m^*,v)\}$ and  $X_{=}:=\{ v\in X: d(m^*,v)=d(m',v)\}.$
Since $m'\in I(m^*,v)$ for all $v\in F_{m^*}(x)$, necessarily $F_{m^*}(x)\subseteq X_{m'}$. Since $|F_{m^*}(x)|>n/2$, this implies that  $|X_{m'}|-|X_{m^*}|>0$.
Since $d(m^*,v)=d(m',v)$ for any $v\in X_{=}$ and since $m^*$ and $m'$ are adjacent, one can easily deduce that   $\Psi_X(m^*)-\Psi_X(m')=|X_{m'}|-|X_{m^*}|>0$, contrary to the assumption
that $m^*$ is a median of $X$ in $H(G)$.  This concludes the proof of Claim 1.

\medskip\noindent
{\bf Claim 2.} For any vertex $x\in X$ and any vertex $y\notin F_{m^*}(x)$, we have $(x|y)_m< 2\delta+1$, i.e., $(x|y)_m\le 2\delta+\frac{1}{2}$.

\medskip\noindent
{\bf Proof of Claim 2.} Recall that $m$ is a vertex of $G$ at distance at most $\delta$ from $m^*$. By definition of $F_{m^*}(x)$ and since $y\notin F_{m^*}(x)$, we obtain
$(x|y)_{m^*}<\delta+1$.  Suppose, by way of contradiction, that  $(x|y)_m\ge 2\delta+1$. This implies that $d(x,m)+d(y,m)-d(x,y)\ge 4\delta+2$, i.e., $(d(x,m)-\delta)+(d(y,m)-\delta)-d(x,y)\ge 2\delta+2$.
Since $d(m,m^*)\le \delta$, by the triangle inequality, we obtain that $d(x,m^*)\ge d(x,m)-\delta$ and $d(y,m^*)\ge d(y,m)-\delta$, yielding $(x|y)_{m^*}=(d(x,m^*)+d(y,m^*)-d(x,y))/2\ge \delta+1$,
contrary to the assumption that $y\notin F_{m^*}(x)$. Hence  $(x|y)_m< 2\delta+1$. Since the Gromov product  $(x|y)_m$ in graphs is an integer or a half-integer,
we obtain $(x|y)_m\le 2\delta+\frac{1}{2}$. This concludes the proof of Claim 2.

\medskip\noindent
{\bf Claim 3.} For any vertex $x\in X$ and any vertex $y\notin F_{m^*}(x)$, any geodesic $[x,y]$ of $G$ intersects the ball $B(m,3\delta+\frac{1}{2})$ of $G$ of
radius $3\delta+\frac{1}{2}$ and center $m$.

\medskip\noindent
{\bf Proof of Claim 3.} Consider a geodesic triangle $\Delta(x,y,m)$ and let $[m,x],[m,y],$ and $[x,y]$ be the sides of this triangle; see Figure~\ref{fig2}(b).  Let $x',y'$ be the points of $[m,x]$ and $[m,y]$, respectively,  located at distance $(x|y)_m$ from $m$. Since the triangles of $G$ are $\delta$-thin, $d(x',y')\le \delta$, moreover $d(x',z')\le \delta$ and $d(y',z')\le \delta$, where $z'$ is the point of $[x,y]$ at distance $(y|m)_x$ from $x$ and at distance $(x|m)_y$ from $y$. Since, by Claim 2, $(x|y)_m\le 2\delta+\frac{1}{2}$, we conclude that $d(m,z')\le d(m,x')+d(x',z')\le 2\delta+\frac{1}{2}+\delta=3\delta+\frac{1}{2}$. This establishes Claim 3.

\medskip
Now, we can conclude the proof of the theorem for graphs with $\delta$-thin triangles. Indeed,  by Claim 1, for any vertex $x$ of $X$, the fiber $F_{m^*}(x)$ contains at most $n/2$ vertices. By Claim~3, the ball $B(m,3\delta+\frac{1}{2})$ intersects any geodesic $[x,y]$ between a vertex $x\in X$ and any vertex $y\notin F_{m^*}(x)$, i.e., $B(m,3\delta+\frac{1}{2})$ intercepts any geodesic between
any $x\in X$ and at least $n/2$ vertices of $X$. This implies that $B(m,3\delta+\frac{1}{2})$ intercepts at least $n^2/4$ of the pairs of vertices of $X$.

Now, suppose that $G$ is a $\delta$-hyperbolic graph. Then the proof is exactly the same except the proof of Claim~3, in which we used $\delta$-thin triangles. We replace Claim 3 by the following assertion:

\medskip\noindent
{\bf Claim~4.} For any vertex $x\in X$ and any vertex $y\notin F_{m^*}(x)$, any geodesic $[x,y]$ of $G$ intersects the ball $B(m,4\delta)$ of $G$ of
radius $4\delta$ and center $m$.

\medskip\noindent
{\bf Proof of Claim~4.} Consider a geodesic triangle $\Delta(x,y,m)$ and let $[m,x],[m,y],$ and $[x,y]$ be the sides of this triangle. Let $z$ be the point of $[x,y]$
at distance $(y|m)_x$ from $x$ and at distance $(x|m)_y$ from $y$. Consider the three distance sums $d(x,y)+d(m,z),d(m,x)+d(z,y),$ and $d(m,y)+d(x,z)$. Notice that $d(x,y)+d(m,z)=(y|m)_x+(x|m)_y+d(m,z),$
$d(m,x)+d(z,y)=(x|y)_m+(y|m)_x+(x|m)_y,$ and $d(m,y)+d(x,z)=(x|y)_m+(x|m)_y+(y|m)_x$. Therefore the distance sums $d(m,x)+d(z,y)$ and $d(m,y)+d(x,z)$ coincide. If these two
sums are the largest distance sums, then from inequality $(y|m)_x+(x|m)_y+d(m,z)\le (x|y)_m+(y|m)_x+(x|m)_y$ and Claim~2 we obtain that $d(m,z)\le (x|y)_m\le 2\delta+\frac{1}{2}$. On the other hand,
if $d(x,y)+d(m,z)$ is the largest distance sum, then, since $G$ is $\delta$-hyperbolic, $d(x,y)+d(m,z)\le d(m,x)+d(z,y)+2\delta$. Consequently,  $d(m,z)\le (x|y)_m+2\delta\le 4\delta+\frac{1}{2}$.
As $d(m,[x,y])$ is an integer, $d(m,[x,y])\leq 4\delta$ holds. This establishes Claim~4.
\end{proof}


In case of $\delta$-hyperbolic Helly graphs, the radius of the intercepting core-ball can be decreased, because in this case $m^*=m$:

\begin{corollary}\label{helly-core} If $G$ is a $\delta$-hyperbolic Helly graph  (respectively, a Helly graph with $\delta$-thin triangles), then any finite subset $X$ of vertices of $G$
has a $(\frac{1}{2},3\delta+\frac{1}{2})$-core (respectively, a $(\frac{1}{2},2\delta)$-core).
\end{corollary}

\begin{remark} The analogue   of Theorem~\ref{core} holds for all geodesic $\delta$-hyperbolic spaces.
\end{remark}

\begin{remark} Contrary to what was asserted in \cite{JoLoBoBa}, for any constant $\alpha,$ the center of mass of a $\delta$-hyperbolic network $G$ can be arbitrarily far from the center of any $(\alpha,O(\delta))$-core of $G$.
\end{remark}
\begin{proof}
Consider the family of trees ($0$-hyperbolic graphs) $T_n,$ $n\in \mathbb{N},$ consisting of a path $P$ on $3\sqrt{n}$ vertices and a star $S$ with $n-3\sqrt{n}$ leaves centered at the end-vertex $x$ of the path $P$. It can be shown by simple computations that  the distance between the center of mass  of $T_n$ and $x$ grows with $n.$ However, any $(\alpha,O(1))$-core must contain the vertex $x.$ Therefore, the distance between the center of the core and the center of the star $S$ is $O(1).$ This implies that the distance between the center of the core and the center of mass of $T_n$ can be made arbitrarily large by taking $n$ large enough.
\end{proof}

\begin{remark} For a finite $n$-vertex $m$-edge ($\delta$-hyperbolic) graph $G=(V,E)$, a $(\frac{1}{2},\rho)$-core with minimum $\rho$ can be found in at most $\rho n O(nm)=O(\rho n^2 m)$ time by iterating over each vertex $v\in V$ and computing the smallest radius $\rho$ such that  $d_G(x,y)<d_{G'}(x,y)$ holds for at least $n^2/4$ pairs $x,y\in V$, where $G'$ is a graph obtained from $G$ by removing the vertices of the ball $B(v,\rho)$. Here, $O(nm)$ stands for the time needed to compute the distance matrices of $G$ and $G'$.
\end{remark}

\begin{remark} The analogue   of Theorem~\ref{core} (with a larger radius of the intercepting ball) holds if the traffic is performed not only along geodesics but also along quasi-geodesics.
A $(\lambda,\epsilon)$-{\it quasi-geodesic} in a metric space $(X,d)$  is a $(\lambda,\epsilon)$-quasi-isometric embedding $c:I\rightarrow X$ (where $I$ is an interval of the real line),
i.e., $\frac{1}{\lambda}|t-t'|-\epsilon\le d(c(t),c(t'))\le \lambda|t-t'|+\epsilon$ holds  for all $t,t'\in I$. By the well-known Morse Lemma \cite[Theorem 1.7, Part III]{BrHa},
there exists $R:=R(\delta,\lambda,\epsilon)$ such that for any quasi-geodesic $c$ and any geodesic segment $[p,q]$ joining the endpoints of $c$, any point of $c$ is at distance at most $R$ from a point of $[p,q]$. Therefore, the ball of radius $R+4\delta$ centered at $m$ will intercept at least one half of the total traffic sent along quasi-geodesics.
\end{remark}

\section{Helly theorem for quasiconvex sets}

The Helly property for balls established in Proposition \ref{helly-balls} of \cite{ChEs} was an important tool in the proof of the existence of cores from the previous section.
In this section, we extend this result to quasiconvex sets of $\delta$-hyperbolic graphs. As a consequence of this result, for any finite collection $\mathcal Q$ of $\epsilon$-quasiconvex sets we obtain a relationship between the packing and transversal numbers of $\mathcal Q$ by showing that  $\tau_{2\epsilon+5\delta}({\mathcal Q})\le \pi({\mathcal Q})$.

\begin{figure}[h]
\begin{center}
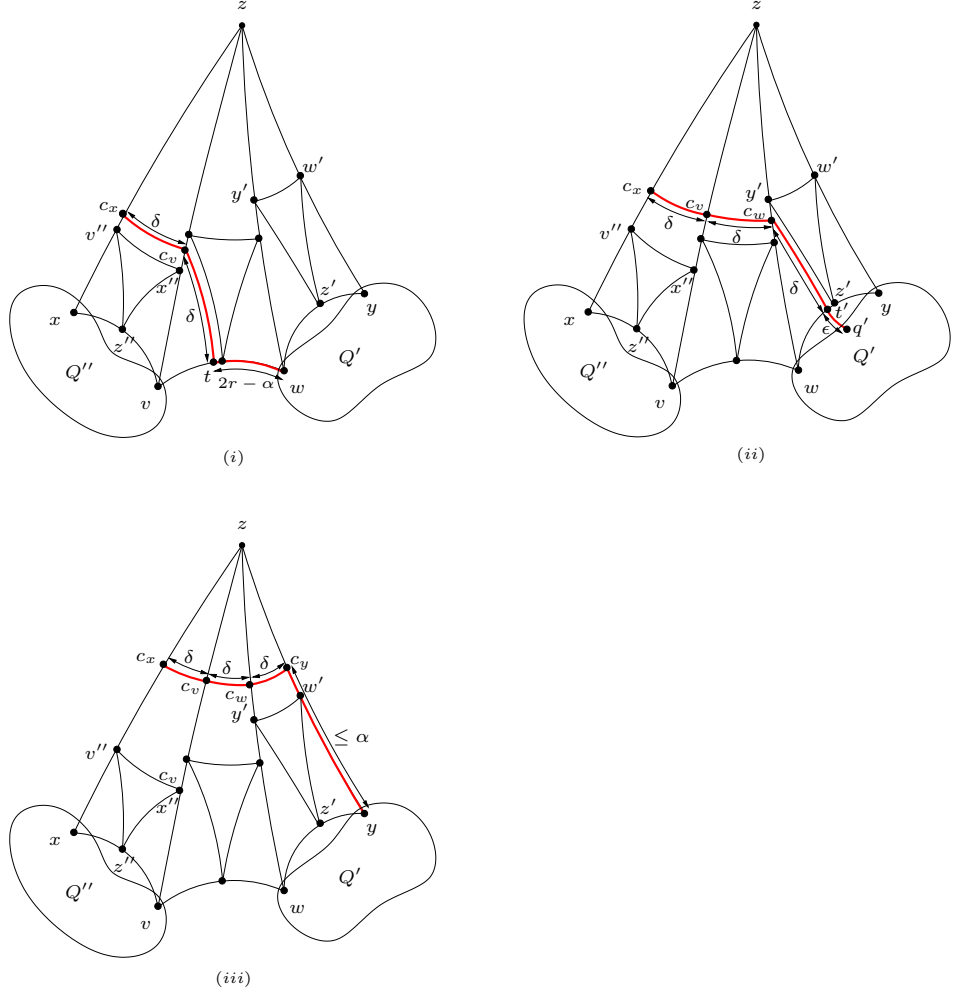
\end{center}
\caption{To the proof of Lemma~\ref{lemma:p&c-2g-close}.} \label{fig4}
\end{figure}

We start with a fundamental lemma, which can be viewed as an extension of \cite[Lemma 1]{ChEs} from balls to all quasiconvex sets. For nonnegative integers $r,\epsilon,\delta$, let $r^*:=r^*(r,\epsilon,\delta):=\max\{ 2\epsilon+5\delta, r+\epsilon+3\delta\}$.

\begin{lemma}  \label{lemma:p&c-2g-close}
Let $G$ be a graph with $\delta$-thin triangles, $r$ be a nonnegative integer, $z$ be a  vertex of $G$, and $Q'$, $Q''$ be two $2r$-close  $\epsilon$-quasiconvex sets  of $G$ such that
$d(z,Q'')\geq d(z,Q')$. 
If $x$ is a vertex of $Q''$ closest to $z$ and $c$ is the vertex at distance $r$ from $x$ on a $(x,z)$-geodesic $[x,z]$, then $d(c,Q')\leq r^*$. If $Q'$ and $Q''$ are two $2r$-close geodesics of $G$, then $d(c,Q')\leq \max\{r+3\delta, 5\delta\}$.
\end{lemma}


\begin{proof}
Let $w\in Q'$ and $v\in Q''$ such that $d(v,w)\leq 2r$. Let $y$ be a vertex of $Q'$ closest to $z$.  Consider two geodesic triangles  $\Delta(z,w,y):=[z,w]\cup [w,y]\cup [y,z]$ and $\Delta(x,v,z):=[x,v]\cup [v,z]\cup [z,x]$, where $[y,z], [z,w], [w,y], [x,v], [z,x]$,  and $[v,z]$ are arbitrary geodesics connecting the corresponding vertices
(if $Q'$  and $Q''$ are geodesics, then we suppose that $[v,x]\subseteq Q''$ and $[w,y]\subseteq Q'$).

Define the points $w'\in [y,z]$ and $z'\in [w,y]$ both located at distance $\alpha_y:=(w|z)_y$ from $y$ and the point $y'\in[w,z]$ located at distance $(w|y)_z$ from $z$ (and hence at distance
$(z|y)_w=d(w,z)-(w|y)_z$ from $w$).
Define the points $v''\in [x,z]$ and $z''\in [v,x]$ both located at distance $\alpha_x:=(v|z)_x$ from $x$ and the point $x''\in[v,z]$ located at distance $(v|x)_z$ from $z$ (and hence at distance
$(z|x)_v=d(v,z)-(v|x)_z$ from $v$).
From the definition of $y',z',w'$, any point $t\in [y',w]\cup [w',y]$ is at distance $\le \delta$ from a point $t'\in [y,z']\cup [z',w]=[y,w]$. Since $y,w\in Q'$ and $Q'$ is $\epsilon$-quasiconvex, $d(t',Q')\le \epsilon$, i.e., there exists a point $q'\in Q'$ such that $d(t',q')\le \epsilon$. Consequently, $d(t,Q')\le d(t,t')+d(t',Q')\le \delta+\epsilon$. Analogously, for any point $t\in [v'',x]\cup [x'',v]$ we have
$d(t,Q'')\le \delta+\epsilon$. In particular, $d(w',Q')\le \delta+\epsilon$ and $d(v'',Q'')\le \delta+\epsilon$. Since $y$ is a vertex of $Q'$ closest to $z$ and $w'\in [z,y]$, $y$ is also a vertex of $Q'$ closest to $w'$, thus $d(w',y)=d(w',Q')$. Consequently, $\alpha_y=d(w',Q')\le \delta+\epsilon$. Analogously, since $x$ is a vertex of $Q''$ closest to $z$ and $v''\in [x,z]$, we deduce that $\alpha_x=d(v'',x)=d(v'',Q'')\le \delta+\epsilon$. If $Q'$ and $Q''$ are geodesics, then one can easily see that $\alpha_x\le \delta$ and $\alpha_y\le \delta$.


Let $\alpha:=\max(r,\alpha_x),$ $\beta:=d(x,z)-\alpha$ and define the points $c_x\in [z,x],$ $c_v\in [z,v],$ $c_w\in [z,w]$ and $c_y\in [z,y]$ located at distance $\beta$ from $z$. Notice that
 $\beta\le (x|v)_z$.  Recall also that
$c$ is the point of $[x,z]$ at distance $r$ from $x$.  Therefore, if $r\ge \alpha_x$, then  $c$  coincides with $c_x\in [v'',z]$ and if $r<\alpha_x$, then $c_x=v''$ and $c$ is located between $x$ and $v''.$ Since $d(x,v'')\le \delta+\epsilon,$ in all cases we deduce that $d(c,c_x)\le \delta+\epsilon.$ Therefore, to bound $d(c,Q')$ it suffices to get a bound on $d(c_x,Q')$. We distinguish between three cases:

\medskip\noindent
{\bf Case 1.} $\beta> (v|w)_z$ (Fig. \ref{fig4}(i)).

\medskip
Since $x$ is a vertex of $Q''$ closest to $z$ and $d(z,c_x)=d(z,c_v)=\beta$, we have $d(v,c_v)\ge d(x,c_x) = \alpha.$ Let $t$ be the point on $[v,w]$ at distance $d(v,c_v)\ge \alpha$ from $v.$ Since $d(v,w)\le 2r$, $d(v,t)\ge \alpha$, and $w\in Q'$,  we deduce that $d(t,Q')\le d(t,w) \le 2r-\alpha.$ Again, since the geodesic triangles in $G$  are $\delta$-thin, we obtain that
$d(c_x,Q')\le d(c_x,c_v)+d(c_v,t)+d(t,Q') \le 2\delta+2r-\alpha.$
If $\alpha=r$, then $c_x$ coincide with $c$ and we get $d(c,Q')\le 2\delta+r.$ Otherwise, $r<\alpha_x$ and consequently $d(c,Q') \le d(c,c_x)+d(c_x,Q') \le \epsilon+\delta + 2\delta+2r-\alpha_x \le \epsilon+3\delta + r.$

\medskip\noindent
{\bf Case 2.} $\beta\le (v|w)_z$ and  $\beta\ge (w|y)_z$ (Fig. \ref{fig4}(ii)).

\medskip
 In this case,  $c_v\in [x'',z]$ and $c_w\in [y',w]$. By what has been shown above, in this case we have $d(c_w,Q')\le \delta+\epsilon$.  Since the geodesic triangles of $G$ are $\delta$-thin, $d(c_x,c_v)\le \delta$ and $d(c_v,c_w)\le \delta$. Consequently,  $d(c_x,Q')\le d(c_x,c_v)+d(c_v,c_w)+d(c_w,Q')\le \epsilon+3\delta$, i.e., in this case $d(c,Q')\le 2\epsilon+4\delta$.

\medskip\noindent
{\bf Case 3.} $\beta\le (v|w)_z$ and  $\beta< (w|y)_z$ (Fig. \ref{fig4}(iii)).

\medskip
In this case, $c_v\in [x'',z]$, while $c_w\in [y',z]$ and $c_y\in [w',z]$.  Since $d(z,x)=d(z,Q'')\ge d(z,Q')=d(z,y)$, $d(z,x)=d(z,c_x)+d(c_x,x)=\beta+\alpha$, and $d(z,y)=d(z,c_y)+d(c_y,y)=\beta+d(c_y,y)$, we conclude that $d(c_y,y)\le \alpha.$ Since the geodesic triangles  in $G$ are $\delta$-thin and $y\in Q'$, we derive that $d(c_x,Q')\le d(c_x,c_v)+d(c_v,c_w)+d(c_w,c_y)+d(c_y,y)\le 3\delta + \alpha.$ If $\alpha=r$, then $c_x$ coincide with $c$ and we get $d(c,Q')\le r+3\delta.$ Otherwise,  $\alpha=\alpha_x\le \epsilon+\delta$ and consequently, $d(c,Q') \le d(c,c_x)+d(c_x,Q')\le \epsilon+\delta + 3\delta + \epsilon+\delta\le 2\epsilon+5\delta.$

\medskip
This proves that in all cases we have $d(c,Q')\le r^*$, where  $r^*:=\max\{ 2\epsilon+5\delta, r+\epsilon+3\delta\}$. If $Q'$ and $Q''$ are two $2r$-close geodesics, then $\alpha_x\le \delta$ and $\alpha_y\le \delta$ and one can see that $d(c,Q')\leq \max\{r+3\delta, 5\delta\}$.
\end{proof}

\begin{theorem} [Helly property for quasiconvex sets] \label{Helly_quasiconvex}  Let $G$ be a graph with  $\delta$-thin triangles  and ${\mathcal Q}$ be a finite collection of
$\epsilon$-quasiconvex subsets of $G$. If the sets of ${\mathcal Q}$ are pairwise $2r$-close, then there exists a  ball $B(c,r^*)$ of radius $r^*$ 
intersecting all sets of $\mathcal Q$. In particular, if ${\mathcal Q}$ is a collection of pairwise intersecting
$\epsilon$-quasiconvex subsets of $G$, then there exists a  ball of radius $2\epsilon+5\delta$ intersecting all sets of $\mathcal Q$ ($r^*=\max\{r+3\delta, 5\delta\}$ if the sets of $\mathcal Q$ are geodesics). 
\end{theorem}


\begin{proof} Let $z$ be an arbitrary vertex of $G$. Suppose that the sets of $\mathcal Q$ are ordered  $Q_1,Q_2,\ldots,Q_n$ in such a way that $d(z,Q_1)\ge d(z,Q_2)\ge\ldots\ge d(z,Q_n)$. Let $x$ be a  vertex of $Q_1$ closest to $z$, i.e., $d(z,x)=d(z,Q_1)$. Let $[z,x]$ be any geodesic between $z$ and $x$  and let $c$ be the point of $[z,x]$ located at distance $r$ from $x$. Since $d(z,Q_1)\ge d(z,Q_i)$ for any $i=2,\ldots,n$, applying for each such $i$  Lemma  \ref{lemma:p&c-2g-close} with $Q'':=Q_1$ and $Q':=Q_i$, we obtain that $d(c,Q_i) \le r^*:=\max\{ 2\epsilon+5\delta, r+\epsilon+3\delta\}$ and $d(x,Q_i)\le r^*:=\max\{r+3\delta, 5\delta\}$ if all sets of $\mathcal Q$ are geodesics. Consequently, $B(c,r^*)$ intersects all sets of $\mathcal Q$.
\end{proof}

\begin{remark} In case when $\mathcal Q$ is a set of geodesics, with a different proof one can slightly improve the bounds in Theorem \ref{Helly_quasiconvex}: $2\delta$ instead of $5\delta$ if the geodesics of $\mathcal Q$ pairwise intersect and $\max\{r+3\delta, 4\delta\}$ instead of $\max\{r+3\delta, 5\delta\}$ if the geodesics of $\mathcal Q$ are $2r$-close.
\end{remark}

\begin{proposition} \label{transversal_packing} Let $\mathcal Q$ be a finite collection of $\epsilon$-quasiconvex sets of a graph $G$ with $\delta$-thin triangles. Then the packing number $\pi(\mathcal Q)$ and the transversal number  $\tau({\mathcal Q}_{2\epsilon+5\delta})$ satisfy the inequality $\tau({\mathcal Q}_{2\epsilon+5\delta})\le \pi({\mathcal Q})$. Moreover,  a hitting set $T$ of ${\mathcal Q}_{2\epsilon+5\delta}$ and a packing ${\mathcal P}$ of ${\mathcal Q}$ such that $|T|=|{\mathcal P}|$ can be constructed in polynomial time. More generally, for any integer $r\ge 0$, $\tau({\mathcal Q}_{r^*})\le \pi({\mathcal Q}_r)$ ($r^*:=\max\{r+3\delta, 5\delta\}$ for geodesics and $r^*:=\max\{r+4\delta, 7\delta\}$ for intervals);  a hitting set $T_{r^*}$ of ${\mathcal Q}_{r^*}$ and a packing ${\mathcal P}_r$ of ${\mathcal Q}_r$ such that $|T_{r^*}|=|{\mathcal P}_r|$ can be constructed in polynomial time.
\end{proposition}


\begin{proof} We start with the first assertion. The proof of this result is algorithmic: we construct
the packing $\mathcal P$ and the hitting set $T$ step by
step ensuring that the following properties hold: (i) each time
when a new point is inserted in $T,$ then a new set of $\mathcal Q$  is also
inserted in ${\mathcal P},$ and (ii) at the end, the sets of $\mathcal P$ are pairwise disjoint
and  $T$ is  a hitting set of ${\mathcal Q}_{2\epsilon+5\delta}$.

The algorithm starts with ${\mathcal Q}^*:={\mathcal Q}$, $T:=\varnothing$, and ${\mathcal P}:=\varnothing$.
Let $z$ be an arbitrary fixed vertex of $G$.
While the set ${\mathcal Q}^*$ is nonempty, the algorithm computes the distances
from $z$ to the sets of ${\mathcal Q}^*$. Suppose that ${\mathcal Q}^*=\{ Q_1,\ldots, Q_n\}$, where
$d(z,Q_1)\ge d(z,Q_2)\ge \ldots \ge d(z,Q_n)$. Set $Q'':=Q_1$. Denote by ${\mathcal Q}'$  the subfamily
of ${\mathcal Q}^*$ consisting of $Q''$ and all sets of ${\mathcal Q}^*$ intersecting the set $Q''$. Let $x$ be a  vertex of $Q''$ closest to
$z$. Applying Lemma \ref{lemma:p&c-2g-close} with $r=0$ we deduce that $d(x,Q')\le 2\epsilon+5\delta$ for
any set $Q'\in {\mathcal Q}'$ (since $r=0$, $x$ plays the role of $c$), i.e, $x\in N_{2\epsilon+5\delta}(Q')$. Then we include
the vertex $x$ in the transversal $T$ and the set $Q'$ in the packing $\mathcal P$, and we update ${\mathcal Q}^*$ by setting
${\mathcal Q}^*\leftarrow {\mathcal Q}^*\setminus {\mathcal Q}'$. By construction, $x$ belongs to all sets of ${\mathcal Q}'_{2\epsilon+5\delta}$
and $Q''$ is disjoint from all sets previously included in $\mathcal P$. This implies that at the end, each set of ${\mathcal Q}_{2\epsilon+5\delta}$
contains a point of $T$ and that the sets of ${\mathcal P}$ are pairwise disjoint. Therefore, $T$ is a transversal of
${\mathcal Q}_{2\epsilon+5\delta}$, ${\mathcal P}$ is a packing of ${\mathcal Q}$, and $|T|=|{\mathcal P}|$.

The second (general) assertion can be established in a similar way subject to the following changes. Initially, we set ${\mathcal Q}^*:={\mathcal Q}$,
$T_{r^*}:=\varnothing$, and ${\mathcal P}_r:=\varnothing$.
At each step, $Q''$ is a furthest from $z$ set  of ${\mathcal Q}^*$, $x$ is a  vertex of $Q''$ closest to $z$ and $c$ is a point on $[z,x]$ at distance
$r$ from $x$.  The set-family ${\mathcal Q}'$ consists of $Q''$ and all sets of ${\mathcal Q}^*$ which are $2r$-close to $Q''$. By Lemma \ref{lemma:p&c-2g-close},
we deduce that $d(x,Q')\le r^*$ for any set $Q'\in {\mathcal Q}'$. Then we include the vertex $c$ in $T_{r^*}$ and the $r$-neighborhood  $N_r(Q'')$ of
$Q''$ in  ${\mathcal P}_r$. Finally, ${\mathcal Q}^*$ is updated by setting ${\mathcal Q}^*:={\mathcal Q}^*\setminus {\mathcal Q}'$. From the construction
it follows that $T_{r^*}$ is a hitting set of ${\mathcal Q}_{r^*}$, ${\mathcal P}_r$ is   a packing of ${\mathcal Q}_r$, and  $|T_{r^*}|=|{\mathcal P}_r|$.
\end{proof}

%


\begin{remark} All results of this section also hold for quasiconvex sets of geodesic $\delta$-hyperbolic spaces.
\end{remark}



\section{Total cores}

In this section, we establish the results about total multi-cores and total beam cores in $\delta$-hyperbolic graphs.
We start with the problem of computation of total $r$-multi-cores of
minimal size. Recall that for a commodity graph $R=(X,F)$ and an integer  $r\ge 0$, $\sigma_r(R)$ denotes the smallest size
of an $r$-multi-core for $R$. Let ${\mathcal I}(R):=\{ I(x,y): xy\in F\}$ denote the set-family in which the sets are the intervals
defined by the edges of $R$. Denote by ${\mathcal I}_r(R)$ the $r$-inflation of ${\mathcal I}(R)$. The following result
shows that for any $r\ge 8\delta$ it is possible to construct in polynomial time a total $r$-multi-core not of optimal size $\sigma_r(R)$ but of size
$\sigma_{r-5\delta}(R).$

\begin{proposition} [Total multi-cores] \label{multicore} Let $G=(V,E)$ be a graph with $\delta$-thin triangles. For any commodity graph $R=(X,F)$
with $X\subseteq V$ and any integer $r\ge 8\delta$,   the following inequalities hold:
$$\pi({\mathcal I}_r(R))\le \tau({\mathcal I}_r(R))\le \sigma_r(R)\le \tau({\mathcal I}_{r-\delta}(R))\le \pi({\mathcal I}_{r-5\delta}(R))\le \sigma_{r-5\delta}(R).$$
A total  $r$-multi-core of $R$ of size  $\sigma_{r-5\delta}(R)$ can be constructed in polynomial time.
\end{proposition}

\begin{proof} Let $C$ be any total $r$-multi-core of the commodity graph $R$ and let $xy$ be any edge of $R$. For any geodesic $[x,y]$ between $x$ and $y$ there exists a vertex $c\in C$  such that
the ball $B(c,r)$ intercepts $[x,y]$, i.e., $[x,y]\cap B(c,r)\ne\varnothing$. Consequently,  $d(c,I(x,y))\le r$ and therefore $C$ is a hitting set of ${\mathcal I}_r(R)$. This shows that $\tau({\mathcal I}_r(R))\le \sigma_r(R)$. The inequality $\pi({\mathcal I}_r(R))\le \tau({\mathcal I}_r(R))$ is trivial.

To prove the inequality $\sigma_r(R)\le \tau({\mathcal I}_{r-\delta}(R))$, let $T$ be a hitting set of ${\mathcal I}_{r-\delta}(R)$. This implies that for any interval $I(x,y)$ with $xy\in F$, the $(r-\delta)$-neighborhood $N_{r-\delta}(I(x,y))$ of $I(x,y)$ intersects $T$. Let $t\in T\cap N_{r-\delta}(I(x,y))$. Let $z$ be a closest to $t$ vertex of $I(x,y)$ and suppose that $d(z,x)=k'$ and $d(z,y)=k''$, where $k'+k''=d(x,y)$.  From the choice of $t$ we conclude that $d(t,z)\le r-\delta$. Since $G$ is a graph with $\delta$-thin triangles, by Lemma \ref{interval_thin} the intervals of $G$ are $\delta$-thin. This implies that if $z'$ is any other vertex of $I(x,y)$ with $d(z',x)=k'$ and $d(z',y)=k''$, then $d(z,z')\le \delta$. Consequently, if $L$ is any geodesic between $x$ and $y$ and $z'$ is a vertex of $L$ at distance $k'$ from $x$, then $d(z,z')\le \delta$, yielding  $d(t,z')\le d(t,z)+d(z,z')\le r$. This implies that the ball $B(t,r)$ intercepts all geodesics $L$ between $x$  and $y$. Consequently, $T$ is an $r$-multi-core for $R$, establishing the inequality $\sigma_r(R)\le \tau({\mathcal I}_{r-\delta}(R))$.


The inequality $\tau({\mathcal I}_{r-\delta}(R))\le \pi({\mathcal I}_{r-5\delta}(R))$ is
obtained by applying Proposition \ref{transversal_packing} for intervals  with $r\ge 8\delta$ (since this proposition is used with $(r-5\delta)$ instead of $r$, we require $r\ge 8\delta$ to ensure $(r-5\delta)+4\delta\ge 7\delta$). Finally, the inequality $\pi({\mathcal I}_{r-5\delta}(R))\le \sigma_{r-5\delta}(R)$ is obtained by
applying the inequality $\pi({\mathcal I}_r(R))\le \sigma_r(R)$ with $r-5\delta$ instead of $r$. By Proposition \ref{transversal_packing}, a hitting set $T_{r-\delta}$ of ${\mathcal I}_{r-\delta}(R)$ and a packing ${\mathcal P}_{r-5\delta}$ of ${\mathcal I}_{r-5\delta}(R)$ such that $|T_{r-\delta}|=|{\mathcal P}_{r-5\delta}|\le \sigma_{r-5\delta}(R)$ can be constructed in polynomial time. Since  $T_{r-\delta}$ is also a total $r$-multi-core for $R$, we are done.
\end{proof}

Now, we will prove the existence of total beam cores in graphs $G$ with $\delta$-thin triangles. We proceed in two stages. First, we show that any two beams of $G$ are $2\delta$-close (a result which may be of independent interest) and we use the Helly theorem for geodesics  to show that $G$ admits a  total beam $5\delta$-core.  Then we present a direct proof that $G$ admits a total  beam $2\delta$-core which can be found  in linear time.

\begin{lemma}  \label{lemma:two-rays}
If  $G$ is a graph with $\delta$-thin triangles, then any two beams of $G$ are $2\delta$-close.
\end{lemma}

\begin{proof} Consider two arbitrary beams $[x,y]$ and $[u,v]$ with $y\in P(x), v\in P(u)$ and two geodesic triangles  $\Delta(u,v,y):=[u,v]\cup [v,y]\cup [y,u]$ and
$\Delta(x,v,y):=[x,y]\cup [v,y]\cup [v,x]$, where $[v,y], [y,u]$ and $[v,x]$ are arbitrary geodesics connecting $v$ with $y$, $y$ with $u$ and
$v$ with $x$.

Let $\alpha:= (y|u)_v= (d(y,v)+d(v,u)-d(y,u))/2$, $\beta:= (v|u)_y=(d(v,y)+d(y,u)-d(v,u))/2=d(v,y)-\alpha$. Since $d(u,v)\geq d(u,y)$, we have
$\alpha-\beta= (d(y,v)+d(v,u)-d(y,u)- d(v,y)-d(y,u)+d(v,u))/2= d(v,u)-d(y,u)\geq 0$, i.e., $\alpha\geq d(v,y)/2\geq \beta$. As geodesic triangles in $G$ are $\delta$-thin, all points of $[y,v]$ at distance at most $\alpha$ from $v$ are at distance at most $\delta$ from the beam $[u,v]$.

Similarly, if $\alpha':= (y|x)_v= (d(y,v)+d(v,x)-d(y,x))/2$ and $\beta':= (v|x)_y=(d(v,y)+d(y,x)-d(v,x))/2=d(v,y)-\alpha'$, then $\alpha'\leq d(v,y)/2\leq \beta'$, since $d(x,v)\leq d(x,y)$. Hence, all points of $[y,v]$ at distance at most $\beta'$ from $y$ are at distance at most $\delta$ from the beam $[x,y]$.

As $\alpha\geq d(v,y)/2$ and $\beta'\geq d(v,y)/2$, there must exist a point in  $[y,v]$ which is at distance at most $\delta$ from both beams $[u,v]$ and $[x,y]$. Hence, the beams  $[u,v]$ and $[x,y]$ are $2\delta$-close, concluding the proof.
\end{proof}

By Lemma \ref{lemma:two-rays}, any two beams of $G$ are $2\delta$-close. Since  any beam is a geodesic, by Theorem \ref{Helly_quasiconvex}  for geodesics, we deduce that all beams of $G$ are
intercepted by a single ball of radius $5\delta$. The following  result improves this:

\begin{proposition} [Existence of total beam cores] \label{thm:beam-core} Let $G$ be a finite graph with $\delta$-thin triangles, $u,v$ be a pair of mutually distant vertices of $G$ and $m$ be a middle vertex of any $(u,v)$-geodesic $[u,v]$. Then the ball $B(m,2\delta)$ is a total beam core of $G$. Moreover, $B(m,2\delta)$ can be computed in linear time if $\delta$ is a constant.
\end{proposition}

\begin{proof}  Let $[x,y]$ be an arbitrary beam of $G$ with $y\in P(x)$ and let  $\Delta(u,v,y):=[u,v]\cup [v,y]\cup [y,u]$ be a geodesic triangle, where $[v,y], [y,u]$ are arbitrary geodesics connecting $y$ with $v$ and $u$.
Let $m_y$ be a point on $[u,v]$ which is at distance $(y|u)_v= (d(y,v)+d(v,u)-d(y,u))/2$ from $v$ and hence at distance $(y|v)_u= (d(y,u)+d(v,u)-d(y,v))/2$ from $u$. Since vertices $u$ and $v$ are mutually distant, we can assume, without loss of generality, that $m$  is located on $[u,v]$ between $v$ and $m_y$, i.e., $d(v,m)\leq d(v,m_y)=(y|u)_v$.

We claim that  $d(v,m)\geq (x|y)_v$. Indeed, if   $d(v,m)<(x|y)_v$, then $d(v,u)/2=d(v,m)<(x|y)_v=(d(x,v)+d(y,v)-d(x,y))/2\leq d(y,v)/2$ since $d(x,v)\leq d(x,y)$ (recall that $y$ is a farthest from $x$ vertex of $G$, i.e., $y\in P(x)$). But then $d(v,u)<d(y,v)$ contradicts with $u\in P(v)$.

Since $\Delta(u,v,y)$ is $\delta$-thin, there is a point $m'$ on $[v,y]$ with $(x|y)_v\leq d(v,m')\leq (y|u)_v$ and $d(m,m')\leq \delta$.
Consider now a geodesic triangle $\Delta(x,v,y):=[x,y]\cup [v,y]\cup [v,x]$, where $[v,x]$ is an arbitrary geodesic connecting $v$ with $x$. As $\Delta(x,v,y)$ is $\delta$-thin and  $(x|y)_v\leq d(v,m')$, we get  $d(m', [x,y])\leq \delta$, i.e., $d(m, [x,y])\leq 2\delta$.

Finally, if $m$ is not a vertex (i.e., it is the middle point of an edge connecting two middle vertices of $[u,v]$), then for any middle vertex $\bar{m}$ of $[u,v]$,  ball $B(\bar{m},2\delta+1/2)=B(\bar{m},2\delta)$ (as the radius of any ball in graphs can be taken as an integer)  of $G$ intercepts all beams of $G$. This establishes the existence of total beam cores.

To compute a total beam core in linear time, recall  that in $\delta$-hyperbolic graphs (and hence in graphs with $\delta$-thin triangles), if $y$ is a most distant vertex from an arbitrary vertex $z$ and $x$ is a most distant vertex from $y$, then $d(x,y)\geq \diam(G)-2\delta$ \cite[Proposition 3]{ChDrEsHaVa}. Hence, using at most $O(\delta)$ {\em breadth-first-searches}, one can generate a sequence of vertices $y:=v_1,x:=v_2, v_3, \dots v_k$ with $k\leq 2\delta+2$  such that each $v_i$ is most distant from $v_{i-1}$ and $v_k$ and $v_{k-1}$ are mutually distant vertices (the initial value $d(x,y)\geq \diam(G)-2\delta$ can be improved at most $2\delta$ times).
\end{proof}

The following proposition shows that, in a graph $G$ with $\delta$-thin triangles, a middle vertex of any geodesic between two mutually distant vertices is not far from the center $C(G)$ of $G$. This shows that, while $(\alpha,r)$-cores are close to median vertices, the total beam $2\delta$-cores are close to center vertices of $G$.  We will need the following lemma of an independent interest.

\begin{lemma} \label{lem:diam-rad}
For every graph with $\delta$-thin triangles, $\diam(G)\geq 2\rad(G)-2\delta-1$ holds.
\end{lemma}

\begin{proof} Assume  $\diam(G)\leq 2\rad(G)-2\delta-2$ and consider a family of balls $\{B(v,\rad(G)-\delta-1): v\in V\}$. Each two balls of that family intersect. By Proposition~\ref{helly-balls}, there must exist a vertex $x$ which is at distance at most $\rad(G)-\delta-1+\delta=\rad(G)-1$ from every vertex $v$ of $G$. Since $\ecc(x)\leq \rad(G)-1$, a contradiction arises.
\end{proof}

\begin{proposition} \label{thm:close-to-center}
Let $G$ be a graph with $\delta$-thin triangles, $u,v$ be a pair of mutually distant vertices of $G$ and $m$ be a middle vertex of any $(u,v)$-geodesic. Then $C(G)\subseteq B(m,4\delta+1)$.
\end{proposition}

\begin{proof} As mentioned in the proof of Proposition \ref{thm:beam-core}, $d(u,v)\geq \diam(G)-2\delta$ holds. Therefore, by Lemma~\ref{lem:diam-rad}, $d(u,v)\geq 2\rad(G)-4\delta-1$, and hence, without loss of generality, we may assume that $d(u,m)\geq \rad(G)-2\delta$ and $d(v,m)\geq \rad(G)-2\delta-1$. Consider an arbitrary vertex $c\in C(G)$ and three distance sums: $S_1=d(c,m)+d(u,v)$,  $S_2=d(u,m)+d(c,v)$, $S_3=d(v,m)+d(c,u)$. If $S_2>S_1$ (analogously, if $S_3>S_1$), then $d(c,m)<d(u,m)+d(c,v)-d(u,v)=d(c,v)-d(v,m)\leq \rad(G)-\rad(G)+2\delta+1=2\delta+1$, i.e., $d(c,m)\leq 2\delta$.

Thus, we may assume that $S_1$ is the largest sum. Without loss of generality, we may assume also that $S_2\geq S_3$ (the case when $S_3\geq S_2$ is similar). As, by Lemma~\ref{hyp_thin}, $G$ is a $\delta$-hyperbolic graph, we have
$2\delta\geq S_1-S_2=d(c,m)+d(u,v)-d(u,m)-d(c,v)=d(c,m)+d(v,m)-d(c,v)\geq d(c,m)+ \rad(G)-2\delta-1-\rad(G)=d(c,m)-2\delta-1$, i.e.,
$d(c,m)\leq 4\delta+1$.
\end{proof}

\section{Hitting sets and packings for $(\kappa,\epsilon)$-quasiconvex sets}

Let $G=(V,E)$ be a graph with $\delta$-thin triangles and  $\kappa>0,\epsilon\ge 0,$ and $r\ge \epsilon+2\delta$ be three nonnegative integers. Set as before $r^*:=r+\epsilon+3\delta$ and let $r':=r^*+\epsilon+3\delta$.  A $(\kappa,\epsilon)$-{\it quasiconvex set} is a collection $\kq_i = \{Q_i^1,Q_i^2,\hdots, Q_i^\kappa\}$ of $\kappa$ (not necessarily disjoint) $\epsilon$-quasiconvex sets of $G.$ Let  $\kqm = \{\kq_1,\kq_2, \hdots \kq_m\}$ be a set-family whose members are $(\kappa,\epsilon)$-quasiconvex sets of $G$. In this section, we establish a relationship between the maximum number $\pi_r(\kqm)$ of pairwise $2r$-apart $(\kappa,\epsilon)$-quasiconvex sets in $\kqm$ and the minimum number $\tau_{r'}(\kqm)$ of balls of radius $r'$ hitting all subsets of $\kqm$.  More precisely, we prove the following result:

\begin{theorem}\label{kappa_quasiconvex} Let   $\kqm=\{ \kq_1, \ldots, \kq_m\}$
be a family of $(\kappa,\epsilon)$-quasiconvex sets of a graph $G=(V,E)$ with $\delta$-thin geodesic triangles. Then $\tau_{r'}(\kqm)\le 2\kappa^2\pi_r(\kqm)$.  Moreover an $r$-packing ${\mathcal P}$ and an $r'$-hitting set $T$ of $\kqm$ such that $|T|\le 2\kappa^2|{\mathcal P}|$ can be constructed in polynomial time.
\end{theorem}

The proof closely follows the proof of Theorem 2 of \cite{ChEs}. Denote by ${\mathcal{Q}}$ the collection of all $\epsilon$-quasiconvex sets participating in the  $(\kappa,\epsilon)$-quasiconvex sets of $\kqm$ (obviously, $|{\mathcal Q}|=\kappa\cdot m$).
For a vertex $v\in V,$ let $\Gamma[v]:=\{i: d(v,\kq_i)\le r\}$ be the set of indices
of all $(\kappa,\epsilon)$-quasiconvex sets $\kq_i$ at distance at most $r$ from $v$.
For any $i=1,\ldots,m,$ let $\Gamma[i]$ be the set of indices of all $(\kappa,\epsilon)$-quasiconvex  sets which cannot be included in a packing containing $\kq_i,$ i.e., $\Gamma[i]=\bigcup \{\Gamma[v]: v \in V, d(v,\kq_i)\le r\}.$ Clearly, if $j\in \Gamma[i],$ then $i\in \Gamma[j].$ Notice also that $i\in \Gamma[i].$

Let $\pi'_r(\kqm)$ and $\tau'_{r}(\kqm)$ be respectively the optima of the following fractional packing and fractional hitting set problems (they can be solved in polynomial time as a pair of dual linear programs):
$$
\begin{array}{lll}
&\left\lbrace
\begin{array}{llll}
    \text{max} & \hspace*{0.1cm}\sum_{i=1}^m x_i & \\
    \text{s.t.}& \hspace*{0.1cm}\sum_{i \in \Gamma[v]} x_i & \le 1 & \hspace*{0.3cm}\forall\, v \in V\\
    &x_i & \ge 0 & \hspace*{0.3cm}\forall\, i=1,\ldots,m\\
\end{array}
\right. &\hspace*{2cm}\Pi_r(\kqm)
\vspace{0.3cm}\\
&\left\lbrace
\begin{array}{llll}
    \text{min} & \hspace*{0.1cm}\sum_{v \in V} y_v &\\
    \text{s.t.}& \hspace*{0.1cm}\sum_{v \in N_{r}(\kq_i)} y_v & \ge 1 & \hspace*{0.3cm}\forall\, i=1,\ldots,m\\
    &y_v & \ge 0 & \hspace*{0.1cm}\hspace*{0.3cm}\forall\, v \in V.\\
\end{array}
\right. &\hspace*{2cm}\Upsilon_{r}(\kqm)
\end{array}
$$

\begin{lemma}
\label{lemma2} If ${\bf x}=\{x_i: i=1,\ldots,m\}$ is an admissible
solution of  $\Pi_{r^*}(\kqm)$, then there exists a
$(\kappa,\epsilon)$-quasiconvex set $\kq_i$ such that $\sum_{j\in \Gamma[i]} x_j\le 2\kappa.$
\end{lemma}

\begin{proof} The proof of this result is inspired by the
averaging  argument used in the proof of Lemma 4.1 of
\cite{BYHaNaShSh}. Define a graph  $\mathbf{\Gamma}$ with $1,\ldots,m$ as
the set of vertices and in which $ij$ is an edge if and only if
$j\in \Gamma[i]$ (and consequently $i\in \Gamma[j]).$ For each edge $ij$ of
$\mathbf \Gamma$, set $z(i,j)=x_i\cdot x_j.$ Since $i\in \Gamma[i],$ define
$z(i,i)=x^2_i.$ In the sum $\sum_{i=1}^m\sum_{j \in \Gamma[i]}z(i,j)$
every $z(i,j)$ is counted twice. On the other hand, an upper bound
on this sum can be obtained in the following way.  Let $t$ be any vertex of $G.$
Pick any  $\epsilon$-quasiconvex set $Q$ in the family $\kqa$.   By Lemma \ref{lemma:p&c-2g-close},
there exists a vertex $c_Q$ at distance at most $r^*$ from every $\epsilon$-quasiconvex set $Q'$ in $\kqa$ that is $2r$-close to $Q$ and such that $d(t,Q')\le d(t,Q)$.
Let $\Gamma^*[c_Q]$ be the set of indices of all $(\kappa,\epsilon)$-quasiconvex sets which are at distance at most $r^*$ from $c_Q.$ Now, for each $(\kappa,\epsilon)$-quasiconvex set $\kq_i$ consider its collection of $\epsilon$-quasiconvex sets, and for each $\epsilon$-quasiconvex set $Q$ in the collection $\kq_i$, add up $z(i,j)$ for all $j\in \Gamma^*[c_Q],$ and then multiply the total sum by 2. This way we
computed the sum $2\sum_{i=1}^m\sum_{Q \in \kq_i}\sum_{j \in
\Gamma^*[c_Q]}z(i,j).$ We assert that this suffices. Indeed, pick
any  $z(i,j)$ for an edge  $ij$ of the graph $\mathbf{\Gamma}.$ Thus the
$(\kappa,\epsilon)$-quasiconvex sets $\kq_i$ and $\kq_j$ contain two $2r$-close $\epsilon$-quasiconvex sets $Q$ and $Q'$.  Suppose without loss of generality that $d(t,Q')\le d(t,Q).$ Then necessarily $j\in \Gamma^*[c_Q]$ because $d(c_Q,Q')\le r^*.$ Hence the term $z(i,j)$ will appear at least once in the triple sum, establishing the required inequality
$$\sum_{i=1}^m\sum_{j \in \Gamma[i]}z(i,j) \le 2\sum_{i=1}^m\sum_{Q \in \kq_i}\sum_{j \in \Gamma^*[c_Q]}z(i,j).$$
Taking into account that $z(i,j)=x_i\cdot x_j=z(j,i),$ this
inequality can be rewritten in the following way:
$$\sum_{i=1}^mx_i\sum_{j \in \Gamma[i]}x_j \le 2\sum_{i=1}^m x_i\sum_{Q \in \kq_i}\sum_{j \in \Gamma^*[c_Q]}x_j.$$
Now, since $c_Q$ is at distance at most $r^*$ from any $\epsilon$-quasiconvex sets in  $\Gamma^*[c_Q]$ and ${\bf x}$ is an admissible solution of
$\Pi_{r^*}(\kqm),$ we conclude that $\sum_{j \in
\Gamma^*[c_Q]}x_j\le 1.$ Thus $\sum_{Q \in \kq_i}\sum_{j \in
\Gamma^*[c_Q]}x_j\le \kappa$ and we deduce
that $\sum_{i=1}^mx_i\sum_{j \in \Gamma[i]}x_j \le 2\kappa\sum_{i=1}^m
x_i.$ Hence, there exists $\kq_i$ such that $x_i\sum_{j\in
\Gamma[i]}x_j\le 2\kappa x_i,$ yielding $\sum_{j\in \Gamma[i]}x_j\le 2\kappa.$
\end{proof}

\begin{lemma} \label{lemma3} It is possible to construct in polynomial time an
integer admissible solution ${\bf x}^*$ of the linear program
$\Pi_r(\kqm)$ of size at least $\pi'_{r^*}(\kqm)/(2\kappa)$.
\end{lemma}

\begin{proof} Let ${\bf x}=\{ x_1,\ldots,x_m\}$ be an optimal (fractional)
solution of the linear program $\Pi_{r^*}(\kqm)$ (it can be
found in polynomial time). We will iteratively use  Lemma
\ref{lemma2} to ${\bf x}$ to derive an integer solution ${\bf
x}^*=\{ x_1^*,\ldots,x_m^*\}$ for the linear program $\Pi_r(\kqm).$ The algorithm starts by setting $\kqm'\leftarrow\kqm$.
By Lemma \ref{lemma2} there exists a
$(\kappa,\epsilon)$-quasiconvex set $\kq_i\in \kqm'$ such that $\sum_{j \in \Gamma[i]}x_{j} \le
2\kappa.$ We set $x^*_i:=1$ and  $x^*_j:=0$ for all $j\in
\Gamma[i]\setminus \{ i\},$ then we remove all $(\kappa,\epsilon)$-quasiconvex sets $\kq_j$
with $j\in \Gamma[i]$ from $\kqm'.$ The algorithm continues with the
current set $\kqm'$ until it becomes empty. Notice
that in all iterations of the algorithm the restriction of $\bf x$
to the current collection $\kqm'$ remains an admissible solution of
the linear program  $\Pi_{r^*}(\kqm')$ defined by $\kqm'.$
This justifies the use of  Lemma \ref{lemma2} in all iterations of
the algorithm.

To show that ${\bf x}^*$ is an admissible solution of $\Pi_r(\kqm)$, suppose by way of contradiction that there exist two $2r$-close $(\kappa,\epsilon)$-quasiconvex sets $\kq_i$ and $\kq_j$ with
$x^*_i=1=x^*_j.$ Suppose that the algorithm selects $\kq_i$ before
$\kq_j.$  Consider the iteration when $x^*_i$ becomes 1. Since $j\in
\Gamma[i],$ at this iteration $x^*_j$ becomes 0 and $\kq_j$ is removed
from $\kqm'.$ Thus $x_j^*$ cannot become 1 at a later stage. This
shows that the $(\kappa,\epsilon)$-quasiconvex sets $\kq_i$ with $x^*_i=1$ indeed
constitute an $r$-packing for $\kqm.$

It remains to compare the costs of the solutions ${\bf x}$ and ${\bf
x}^*.$ For this, notice that according to the algorithm, for each
$(\kappa,\epsilon)$-quasiconvex set $\kq_i$ with $x^*_i=1$ we can define a  subset $\Gamma'[i]$
of $\Gamma[i]$ such that $i\in \Gamma'[i],$ $x^*_j=0$ for all $j\in
\Gamma'[i]\setminus\{ i\},$ and $\sum_{j\in \Gamma'[i]\cup \{ i\}}x_j\le
2\kappa.$ Hence, the $(\kappa,\epsilon)$-quasiconvex sets of $\kqm$ can be
partitioned into groups, such that each group contains a
$(\kappa,\epsilon)$-quasiconvex set selected in the integer solution and the total cost of
the fractional solutions of the sets from each group is at most
$2\kappa.$ This shows that $\sum_{i=1}^m x^*_i\ge
(\sum_{i=1}^mx_i)/(2\kappa).$
\end{proof}

\begin{lemma} \label{lemma4} It is possible to construct in polynomial time an
integer solution ${\bf y}^*$ of the linear program $\Upsilon_{r^*}(\kqm)$ of size at most  $\kappa \pi'_r(\kqm).$
\end{lemma}

\begin{proof} Let ${\bf y}=\{ y_v: v\in V\}$ be an optimal (fractional)
solution of the linear program $\Upsilon_{r}(\kqm).$ Since $\sum_{v
\in \Gamma_r(\kq_i)}y_v\ge 1$ for all $i=1,\ldots,m,$ each $(\kappa,\epsilon)$-quasiconvex set $\kq_i$ contains an $\epsilon$-quasiconvex set, which we will denote by $Q_i$, such that $\kappa\sum_{v \in  N_r(Q_i)}
y_v\ge 1.$  Set $\mathcal{R} := \{ Q_1,\ldots,Q_m\}.$ Notice that
${\bf y}'=\{ y'_v: v\in V\}$ defined by setting  $y'_v=\kappa\cdot
y_v$ if $v\in \bigcup_{i=1}^m Q_i$ and $y'_v=0$ otherwise, is a
fractional $r$-hitting set for the family $\{Q_1,\ldots, Q_m\}.$ Thus
the cost of ${\bf y}'$ is at least $\tau'_r(\mathcal{R})=\pi'_r(\mathcal{R}).$
Notice also that the cost of ${\bf y}'$ is at most $\kappa$ times
the cost of ${\bf y}.$  By Proposition \ref{transversal_packing}, we can
construct in polynomial time a set $T$ of size at most $\pi_r(\mathcal{R})$
which is an $r^*$-hitting set of $\mathcal{R}.$ Let ${\bf y}^*=\{ y_v:
v\in V\}$ be defined by setting $y_v^*:=1$ if $v\in T$ and $y^*_v:=0$
otherwise. Since $\pi_r(\mathcal{R})\le \pi'_r(\mathcal{R}),$ putting all things
together, we obtain:
$$\sum_{v\in V}y^*_v=|T|\le \pi_r(\mathcal{R})\le \pi'_r(\mathcal{R})=\tau'_r(\mathcal{R})\le \sum_{v\in V}
y'_v\le \kappa\sum_{v\in V}y_v=\kappa\tau_r'({\kqm}).$$
\end{proof}

Now, we are ready to complete the proof of Theorem \ref{kappa_quasiconvex}. According to Lemma \ref{lemma3} we can construct in polynomial time an integer  solution ${\bf x}^*$ for $\Pi_r(\kqm)$ of size at least $\pi'_{r^*}(\kqm)/(2\kappa).$ Let ${\mathcal P}=\{ \kq_i: x^*_i=1\}.$ On the other hand, applying Lemma \ref{lemma4} with the radius $r^*$ instead of $r$, we
can construct in polynomial time an integer solution  ${\bf y}^*$ of
the linear program $\Upsilon_{r'}(\kqm)$ of size at most
$\kappa\tau'_{r^*}(\kqm).$ Let $T=\{ v\in V: y^*_v=1\}.$
Since, by duality, $\tau'_{r^*}(\kqm) = \pi'_{r^*}(\kqm)$, we deduce that $|T|\le 2\kappa^2|{\mathcal P}|,$ as required.

\medskip\noindent

{\bf Acknowledgement.} This work has been carried out thanks to the support of ARCHIMEDE LabEx (ANR-11-LABX- 0033) and the A$^*$MIDEX project (ANR-11-IDEX-0001-02) funded by the ``Investissements d'Avenir'' French government program managed by the ANR. F.F.D. thanks the Laboratoire d'Informatique Fondamentale, Aix-Marseille Universit\'e for the hospitality during his visit.

\bibliographystyle{amsalpha}

\end{document}

%% file: delta-thin.pdf_t
\begin{picture}(0,0)%
\includegraphics{delta-thin.pdf}%
\end{picture}%
\setlength{\unitlength}{2279sp}%
\begingroup\makeatletter\ifx\SetFigFont\undefined%
\gdef\SetFigFont#1#2#3#4#5{%
  \reset@font\fontsize{#1}{#2pt}%
  \fontfamily{#3}\fontseries{#4}\fontshape{#5}%
  \selectfont}%
\fi\endgroup%
\begin{picture}(7921,3389)(2863,-5690)
\put(4016,-4831){\makebox(0,0)[b]{\smash{{\SetFigFont{8}{9.6}{\rmdefault}{\mddefault}{\updefault}{\color[rgb]{0,0,0}$\le \delta$}%
}}}}
\put(5111,-4821){\makebox(0,0)[b]{\smash{{\SetFigFont{8}{9.6}{\rmdefault}{\mddefault}{\updefault}{\color[rgb]{0,0,0}$\le \delta$}%
}}}}
\put(9286,-4403){\makebox(0,0)[b]{\smash{{\SetFigFont{8}{9.6}{\rmdefault}{\mddefault}{\updefault}{\color[rgb]{0,0,0}$m$}%
}}}}
\put(6683,-3383){\makebox(0,0)[b]{\smash{{\SetFigFont{8}{9.6}{\rmdefault}{\mddefault}{\updefault}{\color[rgb]{0,0,0}$\varphi$}%
}}}}
\put(5349,-4111){\makebox(0,0)[b]{\smash{{\SetFigFont{8}{9.6}{\rmdefault}{\mddefault}{\updefault}{\color[rgb]{0,0,0}$m_x$}%
}}}}
\put(3774,-4104){\makebox(0,0)[b]{\smash{{\SetFigFont{8}{9.6}{\rmdefault}{\mddefault}{\updefault}{\color[rgb]{0,0,0}$m_z$}%
}}}}
\put(4584,-5191){\makebox(0,0)[b]{\smash{{\SetFigFont{8}{9.6}{\rmdefault}{\mddefault}{\updefault}{\color[rgb]{0,0,0}$m_y$}%
}}}}
\put(2878,-5598){\makebox(0,0)[b]{\smash{{\SetFigFont{8}{9.6}{\rmdefault}{\mddefault}{\updefault}{\color[rgb]{0,0,0}$x$}%
}}}}
\put(6274,-5617){\makebox(0,0)[b]{\smash{{\SetFigFont{8}{9.6}{\rmdefault}{\mddefault}{\updefault}{\color[rgb]{0,0,0}$z$}%
}}}}
\put(7349,-5617){\makebox(0,0)[b]{\smash{{\SetFigFont{8}{9.6}{\rmdefault}{\mddefault}{\updefault}{\color[rgb]{0,0,0}$x$}%
}}}}
\put(10769,-5617){\makebox(0,0)[b]{\smash{{\SetFigFont{8}{9.6}{\rmdefault}{\mddefault}{\updefault}{\color[rgb]{0,0,0}$z$}%
}}}}
\put(4570,-2511){\makebox(0,0)[b]{\smash{{\SetFigFont{8}{9.6}{\rmdefault}{\mddefault}{\updefault}{\color[rgb]{0,0,0}$y$}%
}}}}
\put(9047,-2508){\makebox(0,0)[b]{\smash{{\SetFigFont{8}{9.6}{\rmdefault}{\mddefault}{\updefault}{\color[rgb]{0,0,0}$y$}%
}}}}
\put(4546,-3886){\makebox(0,0)[b]{\smash{{\SetFigFont{8}{9.6}{\rmdefault}{\mddefault}{\updefault}{\color[rgb]{0,0,0}$\le \delta$}%
}}}}
\put(8146,-4876){\makebox(0,0)[rb]{\smash{{\SetFigFont{8}{9.6}{\rmdefault}{\mddefault}{\updefault}{\color[rgb]{0,0,0}$\alpha_x$}%
}}}}
\put(10036,-4921){\makebox(0,0)[lb]{\smash{{\SetFigFont{8}{9.6}{\rmdefault}{\mddefault}{\updefault}{\color[rgb]{0,0,0}$\alpha_z$}%
}}}}
\put(8956,-3616){\makebox(0,0)[rb]{\smash{{\SetFigFont{8}{9.6}{\rmdefault}{\mddefault}{\updefault}{\color[rgb]{0,0,0}$\alpha_y$}%
}}}}
\end{picture}%

%% file: Theorem1.pdf_t
\begin{picture}(0,0)%
\includegraphics{Theorem1.pdf}%
\end{picture}%
\setlength{\unitlength}{2279sp}%
\begingroup\makeatletter\ifx\SetFigFont\undefined%
\gdef\SetFigFont#1#2#3#4#5{%
  \reset@font\fontsize{#1}{#2pt}%
  \fontfamily{#3}\fontseries{#4}\fontshape{#5}%
  \selectfont}%
\fi\endgroup%
\begin{picture}(7580,4646)(-1022,-1094)
\put(897,2302){\makebox(0,0)[b]{\smash{{\SetFigFont{8}{9.6}{\rmdefault}{\mddefault}{\updefault}{\color[rgb]{0,0,0}$x$}%
}}}}
\put(379,-380){\makebox(0,0)[b]{\smash{{\SetFigFont{7}{8.4}{\rmdefault}{\mddefault}{\updefault}{\color[rgb]{0,0,0}$X_{m^*}$}%
}}}}
\put(1015,-73){\makebox(0,0)[b]{\smash{{\SetFigFont{7}{8.4}{\rmdefault}{\mddefault}{\updefault}{\color[rgb]{0,0,0}$m$}%
}}}}
\put(911,448){\makebox(0,0)[lb]{\smash{{\SetFigFont{7}{8.4}{\rmdefault}{\mddefault}{\updefault}{\color[rgb]{0,0,0}$m^*$}%
}}}}
\put(921,1108){\makebox(0,0)[lb]{\smash{{\SetFigFont{7}{8.4}{\rmdefault}{\mddefault}{\updefault}{\color[rgb]{0,0,0}$m'$}%
}}}}
\put(467,1570){\makebox(0,0)[b]{\smash{{\SetFigFont{7}{8.4}{\rmdefault}{\mddefault}{\updefault}{\color[rgb]{0,0,0}$X_{m'}$}%
}}}}
\put(899,3088){\makebox(0,0)[b]{\smash{{\SetFigFont{7}{8.4}{\rmdefault}{\mddefault}{\updefault}{\color[rgb]{0,0,0}$F_{m^*}(x)$}%
}}}}
\put(5491,3172){\makebox(0,0)[b]{\smash{{\SetFigFont{7}{8.4}{\rmdefault}{\mddefault}{\updefault}{\color[rgb]{0,0,0}$F_{m^*}(x)$}%
}}}}
\put(5535,2423){\makebox(0,0)[b]{\smash{{\SetFigFont{7}{8.4}{\rmdefault}{\mddefault}{\updefault}{\color[rgb]{0,0,0}$x$}%
}}}}
\put(3496,1857){\makebox(0,0)[b]{\smash{{\SetFigFont{7}{8.4}{\rmdefault}{\mddefault}{\updefault}{\color[rgb]{0,0,0}$y$}%
}}}}
\put(5325,-76){\makebox(0,0)[b]{\smash{{\SetFigFont{7}{8.4}{\rmdefault}{\mddefault}{\updefault}{\color[rgb]{0,0,0}$m$}%
}}}}
\put(4594,1071){\makebox(0,0)[b]{\smash{{\SetFigFont{7}{8.4}{\rmdefault}{\mddefault}{\updefault}{\color[rgb]{0,0,0}$y'$}%
}}}}
\put(5422,1401){\makebox(0,0)[b]{\smash{{\SetFigFont{7}{8.4}{\rmdefault}{\mddefault}{\updefault}{\color[rgb]{0,0,0}$x'$}%
}}}}
\put(4750,1971){\makebox(0,0)[b]{\smash{{\SetFigFont{7}{8.4}{\rmdefault}{\mddefault}{\updefault}{\color[rgb]{0,0,0}$z'$}%
}}}}
\put(5346,733){\makebox(0,0)[lb]{\smash{{\SetFigFont{7}{8.4}{\rmdefault}{\mddefault}{\updefault}{\color[rgb]{0,0,0}$2\delta+1/2$}%
}}}}
\put(5070,1693){\makebox(0,0)[lb]{\smash{{\SetFigFont{7}{8.4}{\rmdefault}{\mddefault}{\updefault}{\color[rgb]{0,0,0}$\delta$}%
}}}}
\put(-314,794){\makebox(0,0)[b]{\smash{{\SetFigFont{7}{8.4}{\rmdefault}{\mddefault}{\updefault}{\color[rgb]{0,0,0}$X_{=}$}%
}}}}
\put(5399,-1021){\makebox(0,0)[b]{\smash{{\SetFigFont{7}{8.4}{\rmdefault}{\mddefault}{\updefault}{\color[rgb]{0,0,0}(b)}%
}}}}
\put(905,-1021){\makebox(0,0)[b]{\smash{{\SetFigFont{7}{8.4}{\rmdefault}{\mddefault}{\updefault}{\color[rgb]{0,0,0}(a)}%
}}}}
\end{picture}%

%% file: Lemma5-Quasi.pdf_t
\begin{picture}(0,0)%
\includegraphics{Lemma5-Quasi.pdf}%
\end{picture}%
\setlength{\unitlength}{1782sp}%
\begingroup\makeatletter\ifx\SetFigFont\undefined%
\gdef\SetFigFont#1#2#3#4#5{%
  \reset@font\fontsize{#1}{#2pt}%
  \fontfamily{#3}\fontseries{#4}\fontshape{#5}%
  \selectfont}%
\fi\endgroup%
\begin{picture}(13106,13848)(1377,-10043)
\put(11776,3599){\rotatebox{360.0}{\makebox(0,0)[b]{\smash{{\SetFigFont{7}{8.4}{\rmdefault}{\mddefault}{\updefault}{\color[rgb]{0,0,0}$z$}%
}}}}}
\put(13581,-621){\rotatebox{360.0}{\makebox(0,0)[b]{\smash{{\SetFigFont{7}{8.4}{\rmdefault}{\mddefault}{\updefault}{\color[rgb]{0,0,0}$y$}%
}}}}}
\put(9175,-799){\rotatebox{360.0}{\makebox(0,0)[b]{\smash{{\SetFigFont{7}{8.4}{\rmdefault}{\mddefault}{\updefault}{\color[rgb]{0,0,0}$x$}%
}}}}}
\put(10436,-1963){\rotatebox{360.0}{\makebox(0,0)[b]{\smash{{\SetFigFont{7}{8.4}{\rmdefault}{\mddefault}{\updefault}{\color[rgb]{0,0,0}$v$}%
}}}}}
\put(12543,-1746){\rotatebox{360.0}{\makebox(0,0)[b]{\smash{{\SetFigFont{7}{8.4}{\rmdefault}{\mddefault}{\updefault}{\color[rgb]{0,0,0}$w$}%
}}}}}
\put(9926,393){\rotatebox{360.0}{\makebox(0,0)[rb]{\smash{{\SetFigFont{7}{8.4}{\rmdefault}{\mddefault}{\updefault}{\color[rgb]{0,0,0}$v''$}%
}}}}}
\put(10921,-355){\rotatebox{360.0}{\makebox(0,0)[rb]{\smash{{\SetFigFont{7}{8.4}{\rmdefault}{\mddefault}{\updefault}{\color[rgb]{0,0,0}$x''$}%
}}}}}
\put(11876,958){\rotatebox{360.0}{\makebox(0,0)[rb]{\smash{{\SetFigFont{7}{8.4}{\rmdefault}{\mddefault}{\updefault}{\color[rgb]{0,0,0}$y'$}%
}}}}}
\put(12610,1331){\rotatebox{360.0}{\makebox(0,0)[lb]{\smash{{\SetFigFont{7}{8.4}{\rmdefault}{\mddefault}{\updefault}{\color[rgb]{0,0,0}$w'$}%
}}}}}
\put(9518,-1531){\rotatebox{360.0}{\makebox(0,0)[b]{\smash{{\SetFigFont{7}{8.4}{\rmdefault}{\mddefault}{\updefault}{\color[rgb]{0,0,0}$Q''$}%
}}}}}
\put(13283,-1321){\rotatebox{360.0}{\makebox(0,0)[b]{\smash{{\SetFigFont{7}{8.4}{\rmdefault}{\mddefault}{\updefault}{\color[rgb]{0,0,0}$Q'$}%
}}}}}
\put(9978,-1174){\rotatebox{360.0}{\makebox(0,0)[lb]{\smash{{\SetFigFont{7}{8.4}{\rmdefault}{\mddefault}{\updefault}{\color[rgb]{0,0,0}$z''$}%
}}}}}
\put(10204,1062){\rotatebox{360.0}{\makebox(0,0)[rb]{\smash{{\SetFigFont{7}{8.4}{\rmdefault}{\mddefault}{\updefault}{\color[rgb]{0,0,0}$c_x$}%
}}}}}
\put(11574,335){\rotatebox{360.0}{\makebox(0,0)[rb]{\smash{{\SetFigFont{7}{8.4}{\rmdefault}{\mddefault}{\updefault}{\color[rgb]{0,0,0}$\delta$}%
}}}}}
\put(10629,531){\rotatebox{360.0}{\makebox(0,0)[rb]{\smash{{\SetFigFont{7}{8.4}{\rmdefault}{\mddefault}{\updefault}{\color[rgb]{0,0,0}$\delta$}%
}}}}}
\put(12376,-377){\rotatebox{360.0}{\makebox(0,0)[rb]{\smash{{\SetFigFont{7}{8.4}{\rmdefault}{\mddefault}{\updefault}{\color[rgb]{0,0,0}$\delta$}%
}}}}}
\put(13113,-919){\rotatebox{360.0}{\makebox(0,0)[lb]{\smash{{\SetFigFont{7}{8.4}{\rmdefault}{\mddefault}{\updefault}{\color[rgb]{0,0,0}$q'$}%
}}}}}
\put(12803,-932){\rotatebox{360.0}{\makebox(0,0)[rb]{\smash{{\SetFigFont{7}{8.4}{\rmdefault}{\mddefault}{\updefault}{\color[rgb]{0,0,0}$\epsilon$}%
}}}}}
\put(12865,-709){\rotatebox{360.0}{\makebox(0,0)[lb]{\smash{{\SetFigFont{7}{8.4}{\rmdefault}{\mddefault}{\updefault}{\color[rgb]{0,0,0}$t'$}%
}}}}}
\put(12986,-395){\rotatebox{360.0}{\makebox(0,0)[b]{\smash{{\SetFigFont{7}{8.4}{\rmdefault}{\mddefault}{\updefault}{\color[rgb]{0,0,0}$z'$}%
}}}}}
\put(11928,709){\rotatebox{360.0}{\makebox(0,0)[rb]{\smash{{\SetFigFont{7}{8.4}{\rmdefault}{\mddefault}{\updefault}{\color[rgb]{0,0,0}$c_w$}%
}}}}}
\put(11073,829){\rotatebox{360.0}{\makebox(0,0)[rb]{\smash{{\SetFigFont{7}{8.4}{\rmdefault}{\mddefault}{\updefault}{\color[rgb]{0,0,0}$c_v$}%
}}}}}
\put(5455,-5914){\makebox(0,0)[lb]{\smash{{\SetFigFont{7}{8.4}{\rmdefault}{\mddefault}{\updefault}{\color[rgb]{0,0,0}$w'$}%
}}}}
\put(4054,-5898){\makebox(0,0)[rb]{\smash{{\SetFigFont{7}{8.4}{\rmdefault}{\mddefault}{\updefault}{\color[rgb]{0,0,0}$c_v$}%
}}}}
\put(2771,-6852){\makebox(0,0)[rb]{\smash{{\SetFigFont{7}{8.4}{\rmdefault}{\mddefault}{\updefault}{\color[rgb]{0,0,0}$v''$}%
}}}}
\put(4721,-6287){\makebox(0,0)[rb]{\smash{{\SetFigFont{7}{8.4}{\rmdefault}{\mddefault}{\updefault}{\color[rgb]{0,0,0}$y'$}%
}}}}
\put(3746,-7137){\makebox(0,0)[rb]{\smash{{\SetFigFont{7}{8.4}{\rmdefault}{\mddefault}{\updefault}{\color[rgb]{0,0,0}$c_v$}%
}}}}
\put(5889,-6603){\makebox(0,0)[lb]{\smash{{\SetFigFont{7}{8.4}{\rmdefault}{\mddefault}{\updefault}{\color[rgb]{0,0,0}$\le \alpha$}%
}}}}
\put(4722,-6011){\makebox(0,0)[rb]{\smash{{\SetFigFont{7}{8.4}{\rmdefault}{\mddefault}{\updefault}{\color[rgb]{0,0,0}$c_w$}%
}}}}
\put(4621,-3646){\makebox(0,0)[b]{\smash{{\SetFigFont{7}{8.4}{\rmdefault}{\mddefault}{\updefault}{\color[rgb]{0,0,0}$z$}%
}}}}
\put(3439,-5478){\makebox(0,0)[rb]{\smash{{\SetFigFont{7}{8.4}{\rmdefault}{\mddefault}{\updefault}{\color[rgb]{0,0,0}$c_x$}%
}}}}
\put(5004,-5590){\makebox(0,0)[rb]{\smash{{\SetFigFont{7}{8.4}{\rmdefault}{\mddefault}{\updefault}{\color[rgb]{0,0,0}$\delta$}%
}}}}
\put(4480,-5665){\makebox(0,0)[rb]{\smash{{\SetFigFont{7}{8.4}{\rmdefault}{\mddefault}{\updefault}{\color[rgb]{0,0,0}$\delta$}%
}}}}
\put(3953,-5515){\makebox(0,0)[rb]{\smash{{\SetFigFont{7}{8.4}{\rmdefault}{\mddefault}{\updefault}{\color[rgb]{0,0,0}$\delta$}%
}}}}
\put(5283,-5525){\makebox(0,0)[lb]{\smash{{\SetFigFont{7}{8.4}{\rmdefault}{\mddefault}{\updefault}{\color[rgb]{0,0,0}$c_y$}%
}}}}
\put(6426,-7866){\makebox(0,0)[b]{\smash{{\SetFigFont{7}{8.4}{\rmdefault}{\mddefault}{\updefault}{\color[rgb]{0,0,0}$y$}%
}}}}
\put(2020,-8044){\makebox(0,0)[b]{\smash{{\SetFigFont{7}{8.4}{\rmdefault}{\mddefault}{\updefault}{\color[rgb]{0,0,0}$x$}%
}}}}
\put(3281,-9208){\makebox(0,0)[b]{\smash{{\SetFigFont{7}{8.4}{\rmdefault}{\mddefault}{\updefault}{\color[rgb]{0,0,0}$v$}%
}}}}
\put(5388,-8991){\makebox(0,0)[b]{\smash{{\SetFigFont{7}{8.4}{\rmdefault}{\mddefault}{\updefault}{\color[rgb]{0,0,0}$w$}%
}}}}
\put(5823,-7617){\makebox(0,0)[b]{\smash{{\SetFigFont{7}{8.4}{\rmdefault}{\mddefault}{\updefault}{\color[rgb]{0,0,0}$z'$}%
}}}}
\put(3766,-7600){\makebox(0,0)[rb]{\smash{{\SetFigFont{7}{8.4}{\rmdefault}{\mddefault}{\updefault}{\color[rgb]{0,0,0}$x''$}%
}}}}
\put(2363,-8776){\makebox(0,0)[b]{\smash{{\SetFigFont{7}{8.4}{\rmdefault}{\mddefault}{\updefault}{\color[rgb]{0,0,0}$Q''$}%
}}}}
\put(6128,-8566){\makebox(0,0)[b]{\smash{{\SetFigFont{7}{8.4}{\rmdefault}{\mddefault}{\updefault}{\color[rgb]{0,0,0}$Q'$}%
}}}}
\put(2823,-8419){\makebox(0,0)[lb]{\smash{{\SetFigFont{7}{8.4}{\rmdefault}{\mddefault}{\updefault}{\color[rgb]{0,0,0}$z''$}%
}}}}
\put(4501,-9961){\makebox(0,0)[b]{\smash{{\SetFigFont{6}{7.2}{\rmdefault}{\mddefault}{\updefault}{\color[rgb]{0,0,0}$(iii)$}%
}}}}
\put(4621,3591){\makebox(0,0)[b]{\smash{{\SetFigFont{7}{8.4}{\rmdefault}{\mddefault}{\updefault}{\color[rgb]{0,0,0}$z$}%
}}}}
\put(6426,-629){\makebox(0,0)[b]{\smash{{\SetFigFont{7}{8.4}{\rmdefault}{\mddefault}{\updefault}{\color[rgb]{0,0,0}$y$}%
}}}}
\put(2020,-807){\makebox(0,0)[b]{\smash{{\SetFigFont{7}{8.4}{\rmdefault}{\mddefault}{\updefault}{\color[rgb]{0,0,0}$x$}%
}}}}
\put(3281,-1971){\makebox(0,0)[b]{\smash{{\SetFigFont{7}{8.4}{\rmdefault}{\mddefault}{\updefault}{\color[rgb]{0,0,0}$v$}%
}}}}
\put(5388,-1754){\makebox(0,0)[b]{\smash{{\SetFigFont{7}{8.4}{\rmdefault}{\mddefault}{\updefault}{\color[rgb]{0,0,0}$w$}%
}}}}
\put(5823,-380){\makebox(0,0)[b]{\smash{{\SetFigFont{7}{8.4}{\rmdefault}{\mddefault}{\updefault}{\color[rgb]{0,0,0}$z'$}%
}}}}
\put(4195,-1601){\makebox(0,0)[rb]{\smash{{\SetFigFont{7}{8.4}{\rmdefault}{\mddefault}{\updefault}{\color[rgb]{0,0,0}$t$}%
}}}}
\put(2921,799){\makebox(0,0)[rb]{\smash{{\SetFigFont{7}{8.4}{\rmdefault}{\mddefault}{\updefault}{\color[rgb]{0,0,0}$c_x$}%
}}}}
\put(4681,-1681){\makebox(0,0)[b]{\smash{{\SetFigFont{6}{7.2}{\rmdefault}{\mddefault}{\updefault}{\color[rgb]{0,0,0}$2r-\alpha$}%
}}}}
\put(3975,-760){\makebox(0,0)[rb]{\smash{{\SetFigFont{7}{8.4}{\rmdefault}{\mddefault}{\updefault}{\color[rgb]{0,0,0}$\delta$}%
}}}}
\put(3496,561){\makebox(0,0)[rb]{\smash{{\SetFigFont{7}{8.4}{\rmdefault}{\mddefault}{\updefault}{\color[rgb]{0,0,0}$\delta$}%
}}}}
\put(2771,385){\makebox(0,0)[rb]{\smash{{\SetFigFont{7}{8.4}{\rmdefault}{\mddefault}{\updefault}{\color[rgb]{0,0,0}$v''$}%
}}}}
\put(3766,-363){\makebox(0,0)[rb]{\smash{{\SetFigFont{7}{8.4}{\rmdefault}{\mddefault}{\updefault}{\color[rgb]{0,0,0}$x''$}%
}}}}
\put(4721,950){\makebox(0,0)[rb]{\smash{{\SetFigFont{7}{8.4}{\rmdefault}{\mddefault}{\updefault}{\color[rgb]{0,0,0}$y'$}%
}}}}
\put(5455,1323){\makebox(0,0)[lb]{\smash{{\SetFigFont{7}{8.4}{\rmdefault}{\mddefault}{\updefault}{\color[rgb]{0,0,0}$w'$}%
}}}}
\put(3746,100){\makebox(0,0)[rb]{\smash{{\SetFigFont{7}{8.4}{\rmdefault}{\mddefault}{\updefault}{\color[rgb]{0,0,0}$c_v$}%
}}}}
\put(2363,-1539){\makebox(0,0)[b]{\smash{{\SetFigFont{7}{8.4}{\rmdefault}{\mddefault}{\updefault}{\color[rgb]{0,0,0}$Q''$}%
}}}}
\put(6128,-1329){\makebox(0,0)[b]{\smash{{\SetFigFont{7}{8.4}{\rmdefault}{\mddefault}{\updefault}{\color[rgb]{0,0,0}$Q'$}%
}}}}
\put(2823,-1182){\makebox(0,0)[lb]{\smash{{\SetFigFont{7}{8.4}{\rmdefault}{\mddefault}{\updefault}{\color[rgb]{0,0,0}$z''$}%
}}}}
\put(4501,-2716){\makebox(0,0)[b]{\smash{{\SetFigFont{6}{7.2}{\rmdefault}{\mddefault}{\updefault}{\color[rgb]{0,0,0}$(i)$}%
}}}}
\put(11701,-2671){\rotatebox{360.0}{\makebox(0,0)[b]{\smash{{\SetFigFont{6}{7.2}{\rmdefault}{\mddefault}{\updefault}{\color[rgb]{0,0,0}$(ii)$}%
}}}}}
\end{picture}%